\documentclass[11pt,english]{article}

\usepackage[margin=1in]{geometry}
\usepackage{libertine}

\usepackage{amsmath,amsthm,mathtools,amssymb}
\usepackage{thmtools,thm-restate}
\usepackage{algorithm}
\usepackage{algorithmic}%[noend]
\usepackage[shortlabels]{enumitem}
\usepackage[usenames,svgnames,xcdraw,table]{xcolor}
\definecolor{DarkBlue}{rgb}{0.1,0.1,0.5}
\definecolor{DarkGreen}{rgb}{0.1,0.5,0.1}
\usepackage{hyperref}
\hypersetup{
   colorlinks   = true,
 linkcolor    = DarkBlue, % color of internal links
 urlcolor     = DarkBlue, % color of external links
	 citecolor    = DarkGreen % color of links to bibliography
}

\usepackage[capitalize]{cleveref}
\usepackage{graphicx}
\usepackage{svg}
\usepackage{float}
\usepackage{verbatim}
\usepackage{caption}

\usepackage{pgfplots}
\usepackage{pifont}
\newcounter{casenum}

\DeclareMathOperator*{\argmax}{arg\,max}
\DeclareMathOperator*{\argmin}{arg\,min}

\newcommand{\precplus}{\mathop{}\preceq_{++}} 
\newcommand{\mcV}{\mathop{}\mathcal{V}} % used for defining a set of valuations.

\newcommand{\eqx}{\textrm{EQx}}
\newcommand{\eq}{\textrm{EQ}}
\newcommand{\efx}{\textrm{EFx}}
\newcommand{\alloc}{\mathcal{A}}

\newcommand{\cmark}{\ding{51}}%
\newcommand{\xmark}{\ding{55}}%

\newcommand{\lmplus}{\textrm{leximin}$++$}

\newcommand{\NP}{{\rm NP}}
\newcommand{\Part}{\textsc{Partition}}

\newtheorem{theorem}{Theorem}

\newtheorem{definition}{Definition}
\newtheorem{remark}{Remark}

\title{\bfseries Nearly Equitable Allocations Beyond Additivity and Monotonicity}

\author{
Siddharth Barman\thanks{Indian Institute of Science. {\tt barman@iisc.ac.in}} \quad Umang Bhaskar\thanks{Tata Institute of Fundamental Research. {\tt umang@tifr.res.in}} \quad Yeshwant Pandit\thanks{Tata Institute of Fundamental Research. {\tt yeshwant.pandit@tifr.res.in}} \quad Soumyajit Pyne\thanks{Tata Institute of Fundamental Research. {\tt  soumyajit.pyne@tifr.res.in}}}

\date{}

\begin{document}

\maketitle

\begin{abstract}
Equitability (EQ) in fair division requires that items be allocated such that all agents value the bundle they receive equally. With indivisible items, an equitable allocation may not exist, and hence we instead consider a meaningful analog, EQx, that requires equitability up to any item. EQx allocations exist for monotone, additive valuations. However, if (1) the agents' valuations are not additive or (2) the set of indivisible items includes both goods and chores (positively and negatively valued items), then prior to the current work it was not known whether EQx allocations exist or not. 

We study both the existence and efficient computation of EQx allocations. (1) For monotone valuations (not necessarily additive), we show that EQx allocations always exist. Also, for the large class of weakly well-layered valuations, EQx allocations can be found in polynomial time. Further, we prove that approximately EQx allocations can be computed efficiently under general monotone valuations.  (2) For non-monotone valuations, we show that an EQx allocation may not exist, even for two agents with additive valuations. Under some special cases, however, we establish existence and efficient computability of EQx allocations. This includes the case of two agents with additive valuations where each item is either a good or a chore, and there are no mixed items. In addition, we show that, under nonmonotone valuations, determining the existence of EQx allocations is weakly \NP-hard for two agents and strongly \NP-hard for more agents.
\end{abstract}

\section{Introduction}
In the problem of fair division, a central planner (principal) is tasked with \emph{fairly} partitioning a set of items among interested agents. If the items are indivisible, which is our focus, each item must be allocated integrally to an agent. Every agent $i$ has a valuation function $v_i$ that specifies agent $i$'s value for each subset of items. Here, an item $x$ could be a `good', if every agent always values it positively, a `chore', if every agent always values it negatively, or `mixed', if across agents the value for item $x$ can be both positive and negative.

What constitutes a fair allocation of items has no single answer. Over the years, various notions have been studied in depth~\cite{moulin2004fair}. Possibly the most prominent among them are \emph{envy-freeness} and \emph{equitability}. An allocation is said to be envy-free if each agent prefers her own bundle of items to the bundle allocated to anyone else~\cite{foley1966resource}. An allocation is said to be equitable if agents' values for their own bundles are the same and, hence, the agents are equally content~\cite{DubinsS61}. If the agents have identical valuations, then the two notions coincide.

Envy-freeness has received significant attention in classic fair division literature. It is known, for example, that for additive valuations over \emph{divisible} goods, an allocation that maximizes the Nash social welfare (the product of the agents' values for their allocated bundles) is also envy-free \cite{V74equity}. The widely used platform Spliddit.org implements multiple methods to provide solutions for common fair division problems~\cite{GP15spliddit}. The platform uses envy-freeness, in particular, as a fairness criterion for relevant applications.

Notably, equitability is a simpler construct to reason about, since it requires fewer comparisons. To test an allocation for equitability, we only need each agent's value for her own bundle, rather than every agent's value for every bundle.

Perhaps for this reason, equitable solutions are important in practical applications of fair division. Experimental studies have noted that, in specific fair-division settings, users tend to prefer equitable allocations over other notions of fairness ~\cite{HerreinerP09}.\footnote{These prior works refer to equitability as  inequality inversion.} Further, in bargaining experiments, equitability often plays a significant role in determining the outcome~\cite{HerreinerP10}. Also in Spliddit.org, for fairly dividing rent among housemates, equitability was noted to be important as a refining  criterion, following envy-freeness. The latest rent-division algorithm used in Spliddit.org computes solutions that satisfy this supporting objective~\cite{GalMPZ17}.\footnote{Specifically, rent divisions that maximize the minimum value, called maximin solutions, are nearly equitable. The latest algorithm implemented in Spliddit.org finds rent divisions that satisfy envy-freeness and are also maximin.}

The real-world significance of equitability is further highlighted by considering divorce settlements. The two legal means of dividing property in the United States are \emph{community property} and \emph{equitable distribution}~\cite{Divorce}. In the community property rule, holdings are divided equally among the divorcing couple and, hence, the rule induces an envy-free division. Equitable distribution, on the other hand, takes into account various factors (such as the employability and financial needs of each party) for dividing the assets, and, in particular, makes inter-personal comparisons of value. Most states in the US follow equitable distribution, i.e., the courts divide assets and liabilities based on equitability. The definition of equitability in this situation is somewhat intuitive. However, the evidence clearly suggests that equitability is an important concept in practical situations. 

Motivated by such considerations, the current work studies equitability (\eq), with the focus on allocating indivisible items. In the discrete fair division context, simple examples (with two agents and a single indivisible item) demonstrate that exactly equitable allocations may not exist. Hence, recent works have focused on relaxations. A compellingly strong relaxation is obtained by requiring that any existing inequality in agents' values is switched by the (hypothetical) removal of any good or chore in an appropriate manner. This relaxation is called equitability up to any item, \eqx{}. Specifically, an allocation is said to be \eqx{} if, whenever agent $i$ has a lower value for her bundle, say $A_i$, than some other agent $j$ for her bundle, say $A_j$ (i.e., $v_i(A_i) < v_j(A_j)$), the removal of any positively-valued item (good) from $A_j$ or the removal of any negatively-valued item (chore) from $A_i$ ensures that the inequality is (weakly) reversed. Similarly, an allocation is \efx{} if, whenever agent $i$ has higher value for agent $j$'s bundle than her own (i.e., $v_i(A_i) < v_i(A_j)$), this preference can be weakly reversed by removing a good from $A_j$ or a chore from $A_i$.

Given the relevance of equitability, a fundamental question in discrete fair division is whether \eqx{} allocations always exist. For the specific case of additively valued goods, \eqx{} allocations are known to exist; this result is obtained via a greedy algorithm~\cite{GMT14near}.\footnote{\eqx{} is termed as \emph{near jealousy-freeness} in \cite{GMT14near}.}  Beyond this setting, however, this question has not been addressed in the literature. Our work addresses this notable gap. In particular, moving beyond monotone additive valuations, the current work establishes novel existential and algorithmic guarantees for \eqx{} allocations. For general monotone valuations, for example, our work shows the universal existence of \eqx{} allocations.

\subsection{Our Results and Techniques} 
\begin{table*}[t]
\centering
\begin{tabular}{ | p{3.4cm} || p{5.5cm}| p{6cm} | }
\hline 
 & \qquad \qquad Existence & \qquad \quad Polynomial-Time Algorithm  \\
 \hline \hline
 Monotone, Beyond &   {\large \cmark} Theorem \ref{thm:Monotone} &   {\large \cmark} Weakly Well-Layered: Theorem \ref{thm:WWLPolyTime} \\ 
Additive (Section \ref{sec:Monotone})  & &  {\large \cmark} Approximate \eqx{}: Theorem \ref{thm:MonotoneApproximation} \\
& & \\
 \hline
Nonmonotone Additive (Section \ref{sec:Chores}) &   {\large \xmark}  Subjective Valuations: Theorem \ref{theorem:no-mas-eqx} &   {\large \cmark} Two agents, Objective Valuations:   Theorem~\ref{theorem:two-agents} \\ 
 \ &   {\large \cmark} Identical Chores: Theorem \ref{thm:LeximinIsEQx} & {\large \xmark} Two agents, Subjective Valuations: Theorem~\ref{theorem:no-mas-eqx} (weakly \NP-hard)\\ 
 \ &  {\large \cmark} Single Chore: Theorem~\ref{thm:SingleChore}  & {\large \cmark} Constant number of agents: Theorem~\ref{theorem:PseudoPoly} (pseudo-polynomial time algorithm)\\
& & {\large \xmark} Arbitrary number of agents: Theorem~\ref{theorem:strongNP-hard} (strongly \NP-hard)\\
\hline
\end{tabular}
\caption{Our Results}
\label{table:results}
\end{table*}

Under monotone, nondecreasing valuations, we establish (in Section \ref{sec:Monotone}): 

\begin{itemize}
\item \eqx{} allocations always exist, and can be computed in pseudo-polynomial time (Theorem \ref{thm:Monotone}). Our algorithm is based on a modification of an Add-and-Fix algorithm, proposed earlier for \efx{} under \emph{identical} monotone nonincreasing valuations~\cite{BarmanNV23}.

\item Under \emph{weakly well-layered} valuations, \eqx{} allocations can be computed in polynomial time (Theorem \ref{thm:WWLPolyTime}). Weakly well-layered functions are an encompassing and natural class of valuations that include gross substitutes, weighted matroid-rank functions, budget-additive, well-layered, and cancelable valuations~\cite{GoldbergHH23}. 

\item For any $\varepsilon \in (0,1)$, a $(1 - \varepsilon)$-approximately \eqx{} allocation can be computed in time polynomial in $1/\varepsilon$ and the input size (Theorem \ref{thm:MonotoneApproximation}). 
\end{itemize}

Finding an \efx{} allocation is known to be hard~\cite{PR20almost}. Since the negative result holds even for two agents, with {\it identical} submodular valuations, the hardness also applies to \eqx{}. This observation signifies that our polynomial-time algorithm for weakly well-layered valuations is the best possible, in the sense that such a positive result is unlikely for submodular valuations in general. 

We then address nonmonotone additive valuations (Section \ref{sec:Chores}). In comparison to monotone valuations, the results here are mixed and highlight a complicated landscape.
\begin{itemize}
\item Existential guarantees: 
\begin{itemize}
\item For agents with subjective valuations---when mixed items are allowed---\eqx{} allocations may not exist, even for just two agents with normalized valuations (Theorem \ref{theorem:no-mas-eqx}).\footnote{Valuations are \emph{normalised} if, for the set $M$ of all the goods,  $v_i(M)$ is equal for all agents $i$.} This negative result necessitates focusing on objective valuations, wherein each item is exclusively a good or a chore. 

\item For $n$ agents with additive objective valuations where each chore $c$ has the same value for the agents (i.e., $v_i(c) = v_j(c)$ for all agents $i$ and $j$), we show---using the \lmplus{} ordering---that \eqx{} allocations always exist (Theorem \ref{thm:LeximinIsEQx}). 

\item For $n$ agents with additive objective valuations and a single chore, we show that \eqx{} allocations exist (Theorem~\ref{thm:SingleChore}).
\end{itemize}

The last result (Theorem~\ref{thm:SingleChore}) is technically challenging. It is obtained via a local search algorithm, that resolves \eqx{} violations, whenever they arise, in a specific order. The analysis is subtle and needs to keep track of multiple progress measures, including the \lmplus{} value.

\item Computational results: \begin{itemize}
\item For two agents with additive objective valuations, we develop a polynomial-time algorithm for finding \eqx{} allocations (Theorem~\ref{theorem:two-agents}). We note that for agents with \emph{identical} additive valuations, Aziz and Rey \cite{AzizR20} develop an efficient algorithm for finding \efx{} (and, hence, \eqx{}) allocations. By the well-known cut-and-choose protocol, this gives an \efx{} algorithm for two nonidentical agents. However, cut-and-choose does not work for \eqx{} and, hence, prior to our result there was no known algorithm for finding \eqx{} allocation among two nonidentical agents.  

Many instances of fair division, in fact, consist of just two agents (such as in divorce settlement, or in the experiments of Gal et al.~\cite{GalMPZ17}), hence our result for two agents is of practical significance.

\item For constant number of agents with nonmonotone, additive, subjective valuations, we provide a pseudo-polynomial time algorithm (Theorem \ref{theorem:PseudoPoly}).

\item For two agents with nonmonotone, additive, subjective valuations it is weakly \NP-hard to determine whether an \eqx{} allocation exists or not (Theorem \ref{theorem:no-mas-eqx}).

\item For arbitrary number of agents with nonmonotone, additive, subjective valuations determining whether there exists an \eqx{} allocation is strongly \NP-hard (Theorem \ref{theorem:strongNP-hard}).

\end{itemize}
\end{itemize}

Table \ref{table:results} summarizes our results. The existence of \eqx{} allocation under objective, additive valuations stands as an interesting, open question.

\subsection{Additional Related Work} 
Equitability was first studied in the divisible setting of cake cutting. Here, equitable allocations for $n$ agents exist ~\cite{DubinsS61,CechlarovaDP13,Cheze17} and, in particular, do not require additivity~\cite{AumannD15}. However, no finite algorithm can find an exactly equitable allocation~\cite{ProcacciaW17}, though approximately equitable allocations can be computed in near-linear time~\cite{CechlarovaP12}.

For indivisible items, \eqx{} allocations were first shown to exist for monotone additive valuations~\cite{GourvesMT14,FSV+20equitable}. This existential guarantee was obtained via an efficient greedy algorithm. For additive valuations that are strictly positive, allocations that are both \eqx{} and Pareto optimal are known to exist~\cite{FSV+19equitable}.

When the agents have identical valuations, \efx{} and \eqx{} allocations coincide. Hence, under identical valuations, existential guarantees obtained for \efx{} allocations extend to \eqx{} as well. In particular, for identical monotone (nondecreasing) valuations, \efx{} allocations were shown to exist using the \lmplus{} construct  ~\cite{PR20almost}. This work also showed that even for two agents with identical submodular valuations, finding an \efx{} allocation requires an exponential number of queries. This problem is also {\rm PLS}-complete~\cite{GoldbergHH23}. For objective identical valuations, the existence of \efx{} allocations was shown through a modification of the leximin construct~\cite{ChenL20}. For additive identical valuations, \efx{} existence, as well as efficient computation, was obtained by Aziz and Rey~\cite{AzizR20}. %via the Egal-Sequential algorithm

Finally, a number of recent papers have also studied the loss of efficiency for equitable and near equitable allocations~\cite{CKK+12efficiency,AumannD15,FSV+19equitable,FSV+20equitable,SunCD23,BhaskarMSV23}.

\section{Notation and Preliminaries}
A fair division instance $(N,M,\mcV)$ consists of a set $N = \{1, 2, \ldots, n\}$ of agents, a set $M$ of indivisible items, with $m = |M|$, and a valuation function $v_i \in \mcV$ for each agent $i \in N$. The valuation $v_i: 2^M \rightarrow \mathbb{Z}$ specifies agent $i$'s value for every subset of the items. We will assume, throughout, that $v_i(\emptyset) = 0$, and all the agents' values are integral.\footnote{The integrality assumption holds without loss of generality for rational values, since we can multiply the rational numbers by the product of their denominators to get integral ones. Note that under such a scaling, the size of the input only increases polynomially.}

For notational convenience, for single items $x \in M$, we use $v_i(x)$ and $v_i(S\cup x)$ to denote $v_i(\{x\})$ and $v_i(S \cup \{x\})$, respectively. Let $V_{\max} := \max_{i \in [n]} v_i(M)$. Given a valuation $v$ and a subset $S \subseteq M$ of items, we say an item $x \in S$ is a \emph{good} if $v(S \cup x) \ge v(S)$, and is a \emph{chore} otherwise. Note that whether an item is a good or a chore depends on the current allocation. If $v(S \cup x) \ge v(S)$ for all $S \subseteq M$, we will say that item $x$ is a good (or a chore, if all inequalities are weakly reversed and some inequality is strict).

A valuation $v$ is monotone nondecreasing if $v(S \cup x) \geq v(S)$ for all subsets $S \subseteq M$ and all items $x \in M$. In this case, all items are goods for this valuation. Similarly, valuation $v$ is monotone nonincreasing if $v(S \cup x) \le v(S)$ for all $S \subseteq M$ and $x \in M$ (and all items are chores in this case). Function $v$ is additive if $v(S) = \sum_{x \in S} v(x)$, for all subsets $S \subseteq M$. We will say that a fair division instance has objective valuations if for each item $x \in M$ (i) either $v_i(S \cup x) \ge v_i(S)$ for all agents $i$ and subsets $S \subseteq M$, (ii) or $v_i(S \cup x) \le v_i(S)$ for all agents $i$ and subsets $S \subseteq M$. Under (i), the item $x$ is a good, and if case (ii) holds (and the inequality is strict for some agent $i$ and subset $S$), then item $x$ is a chore. We will mainly address objective valuations and, hence, every item is unequivocally either a good or a chore.\footnote{The key exception here is the example given in Theorem~\ref{theorem:no-mas-eqx}, which shows that if the valuations are not objective, then \eqx{} allocations may fail to exist.} We will typically use $g$ to denote a good  and $c$ to denote a chore. Under objective valuations $v_i$, for each good $g$ the value $v_i(g) \geq 0$ and for each chore $c$ we have $v_i(c) \leq 0$.

An allocation $\alloc := (A_1, \ldots, A_n)$ is a partition of items $M$ into $n$ pairwise disjoint subsets. Here, subset of items $A_i \subseteq M$ is assigned to agent $i$ (also called agent $i$'s bundle). At times (such as when analyzing the interim allocations obtained by our algorithms), we may also consider partial allocations wherein not all items are assigned among the agents, i.e., for the pairwise disjoint bundles we have $\cup_{i=1}^n A_i \subsetneq M$. Given an allocation $\alloc$, we say an agent $p$ is poorest if $v_p(A_p) = \min_{i \in N} v_i(A_i)$, and agent $r$ is richest if $v_r(A_r) = \max_{i \in N} v_i(A_i)$.  

\paragraph*{Equitability and Envy-Freeness.} An allocation $\alloc$ is equitable if $v_i(A_i) = v_j(A_j)$ for all agents $i, j \in N$. An allocation is said to be envy-free if $v_i(A_i) \ge v_i(A_j)$ for all agents $i, j \in N$. Hence, equity requires that all agents have equal value, while envy-freeness requires that each agent values her bundle more than that of any other agent. As mentioned, even in simple instances (with a single indivisible good and two agents), both equitable and envy-free allocations do not exist. Hence, we consider relaxations of these notions.

Specifically, an allocation $\alloc=(A_1, \ldots, A_n)$ is said to be equitable up to any item (\eqx{}) if \\
\noindent 
({\tt 1}) For every pair of agents $i$, $j$ and for each good $g \in A_j$, we have $v_j(A_j \setminus \{ g \}) \leq v_i(A_i)$, and \\
({\tt 2}) For every pair of agents $i$, $j$ and for each chore $c \in A_i$, we have $v_i(A_i \setminus\{c\}) \geq v_j(A_j)$. 

That is, in an \eqx{} allocation, for each agent, the removal of any good assigned to her makes her a poorest agent, and the removal of any chore assigned to her must make her a richest agent. \\

Allocation $\alloc=(A_1, \ldots, A_n)$ is said to be envy-free up to any item (\efx{}) if, for all agents $i, j \in N$, we have $v_i(A_i) \ge v_i(A_j  \setminus \{ g\})$ for each good $g \in A_j$ and $v_i(A_i \setminus \{ c\}) \ge v_i(A_j)$ for each chore $c \in A_i$. Hence, from the perspective of agent $i$, the removal of any good from another agent's bundle should make it less valuable than $i$'s own bundle, and the removal of any chore from agent $i$'s bundle should make it more valuable than any other bundle.

Our results for monotone nondecreasing valuations (when all items are goods) are detailed in Section~\ref{sec:Monotone}. Section~\ref{sec:Chores} presents our results for nonmonotone additive valuations.

\section{Monotone Valuations}
\label{sec:Monotone}

%We first present our results for monotone valuations, where for any subset of items $S \subseteq M$ and item $x \in M$, the value $v_i(S \cup x) \ge v_i(S)$. 
In this section, all items have nonnegative marginal values. That is, $v_i(S \cup x) \ge v_i(S)$ for all items $x \in M$, agents $i \in N$, and subsets $S \subseteq M$. Hence, all items are goods in this section. We assume a standard value-oracle access to the valuations, that given any agent $i \in N$ and subset $S \subseteq M$, returns $v_i(S) \in \mathbb{Z}_{\ge 0}$. %As stated previously, we assume that all the values are nonnegative integers. 

For monotone valuations, we establish strong positive results towards the existence of \eqx{} allocations. Our primary result shows that, for general monotone valuations, \eqx{} allocations \emph{always} exist, and for a broad class of valuations---termed \emph{weakly well-layered} valuations---such fair allocations can be found in polynomial time. 

\begin{definition}[\cite{GoldbergHH23}]
\label{definition:well-layered}
    A valuation function $v:2^M \rightarrow \mathbb{Z}_{\ge 0}$ is said to be \emph{weakly well-layered} if for any set $M' \subseteq M$ the sets $S_0$, $S_1$, $\ldots$ obtained by the greedy algorithm (that is, $S_0 = \emptyset$ and $S_i = S_{i-1} \cup \{x_i\}$, where $x_i \in \argmax_{x \in M' \setminus S_{i-1}} v(S_{i-1} \cup x)$, for $i \le |M'|$) are optimal, in the sense that $v(S_i) = \max_{S \subseteq M': |S| = i} v(S)$ for all $i \le |M'|$.
\end{definition}

Weakly well-layered valuations intuitively capture valuations where optimal sets can be obtained via a greedy algorithm. Several interesting and widely-studied classes of valuations are weakly well-layered, including gross substitutes (which include weighted matroid rank functions), budget-additive, well-layered, and cancelable valuations; Figure 1 in ~\cite{GoldbergHH23} is helpful in visualizing the relation between these classes. Notably, submodular functions are {\it not} weakly well-layered. This is an affirming observation, since it is known that, even for two agents with identical submodular valuations, obtaining an \efx{} (and, hence, \eqx{}) allocation is {\rm PLS}-complete \emph{and} requires exponentially many value queries. 

%\begin{proposition}[Theorem 4.1~\cite{GoldbergHH23}, Theorem 3.10~\cite{PlautR20}]
%The problem of computing an \efx{} allocation, even for two identical agents with submodular valuations, is PLS-complete and has exponential query complexity.
%\end{proposition}

%\UB{Explain PLS.}

Our results are obtained through a modification and careful analysis of the Add-and-Fix algorithm, earlier used for obtaining \efx{} under {\it identical}  cost functions (i.e., when all items are chores)~\cite{BarmanNV23}. %For costs that additionally have binary marginals, Add-and-Fix was shown to run in polynomial time. 

The modified version of Add-and-Fix uses a greedy selection criterion; see Algorithm \ref{alg:AddandFix2}.  In each iteration, the algorithm identifies a poorest agent $p \in N$ in the current allocation. Then, agent $p$ selects goods greedily from the unassigned ones. That is, from among the unassigned goods, $p$ iteratively selects goods with \emph{maximum marginal value}, until it is no longer the poorest agent. This is the Add phase in the algorithm (Lines \ref{line:AddCondition} to \ref{line:add-end}). If after adding these goods, the allocation obtained is {not} \eqx{}, this must be because of the goods assigned to agent $p$. In the Fix phase (Lines \ref{line:FixCondition} to \ref{line:fix-end}), violating goods are iteratively removed from agent $p$'s bundle, until the allocation is \eqx{}. 

After each iteration, either the value of agent $p$ increases, or the algorithm terminates (Claim \ref{claim:AddandFixIteration}). We differ from the original Add-and-Fix in two aspects: the original algorithm chose a richest agent in each iteration, since it dealt with chores, whereas we select a poorest agent. Secondly, and crucially for our results for efficient computation, the original algorithm assigned an arbitrary chore to the richest agent, while we select goods with maximum marginal contribution to the poorest agent. Also, note that the \efx{} guarantee in \cite{BarmanNV23} is obtained for identical costs functions, whereas the \eqx{} result here holds for (monotone) non-identical valuations. 

\begin{algorithm}[ht]
\caption{Greedy Add-and-Fix}
\label{alg:AddandFix2}
\textbf{Input}: Fair division instance $(N,M,\mcV)$ with value-oracle access to monotone, nondecreasing valuations.\\
\textbf{Output}: \eqx{} allocation $\alloc$.% $=(A_1, \ldots, A_n)$
\begin{algorithmic}[1] %[1] enables line numbers
\STATE Initialize bundles $A_i = \emptyset$ for all agents $i$, and initialize $U = M$ as the set of unassigned goods.
\WHILE{$U \neq \emptyset$}
	\STATE Let ${\displaystyle p \in \argmin_{i \in N}  v_i(A_i)}$ and ${\displaystyle p' \in \argmin_{i \in N \setminus \{ p \}}  v_i(A_i)}$. \\
	\COMMENT{$p$ and $p'$ are a `poorest' and 'second poorest' agent, respectively.} \\ 
	\COMMENT{{Add Phase}: Lines \ref{line:AddCondition} to \ref{line:add-end}}
	\WHILE{$v_p(A_p) \le v_{p'}(A_{p'})$ \AND $U \neq \emptyset$} \label{line:AddCondition}
		\STATE Let $g^* \in \argmax_{g \in U} (v_p(A_p \cup g) - v_p(A_p))$.
		\STATE Update $A_p \gets A_p \cup \{g^*\}$ and $U \gets U \setminus \{g^*\}$. %\label{line:add-end}
	\ENDWHILE \label{line:add-end}
	
	\COMMENT{{Fix Phase}: Lines \ref{line:FixCondition} to \ref{line:fix-end}}
	\WHILE{there exists $\widehat{g} \in A_p$ such that $v_p(A_p \setminus \{ \widehat{g} \}) > v_{p'}(A_{p'})$} \label{line:FixCondition}
		\STATE Update $A_p \gets A_p \setminus \{\widehat{g}\}$ and $U \gets U \cup \{\widehat{g}\}$. 
	\ENDWHILE \label{line:fix-end}
\ENDWHILE
\RETURN {Allocation $\alloc = (A_1, \ldots, A_n)$}.
\end{algorithmic}
\end{algorithm}

In our proofs, an iteration of the outer while-loop is called an outer iteration. Note that in every outer iteration, the allocation to every agent, other than $p$, remains unchanged. We will use the following claim. %, which we prove in Appendix~\ref{appendix:monotone}.

\begin{restatable}{claim}{ClaimAddandFixIt}
After each outer iteration, (i) either the value of the selected agent $p$ strictly increases such that $p$ is no longer the poorest agent (and the values of the other agents remain unchanged), or (ii) all the remaining unassigned goods are allocated to agent $p$ and the algorithm terminates.
\label{claim:AddandFixIteration}
\end{restatable}
\begin{proof}
At the beginning of the outer iteration, $v_p(A_p) \le v_{p'}(A_{p'})$. The Add phase stops only when either all remaining goods are allocated to agent $p$ or $v_p(A_p) > v_{p'}(A_{p'})$. Hence, at the end of the Add phase either $v_p(A_p) > v_{p'}(A_{p'})$, or all remaining goods are assigned to agent $p$ and $v_p(A_p) \le v_{p'}(A_{p})$. In the latter case, the Fix phase is not executed, and the algorithm terminates as claimed. In the former case, the Fix phase is executed. Note, however, that a good $\widehat{g}$ is removed from agent $p$ only if $v_p(A_p \setminus \{ \widehat{g} \}) > v_{p'}(A_{p'})$. Hence, after termination of the Fix phase, we continue to have $v_p(A_p) > v_{p'}(A_{p'})$. Therefore, the value of agent $p$ has strictly increased to above that of agent $p'$, as claimed.
\end{proof}

\begin{theorem}
Given any fair division instance with monotone valuations, Algorithm~\ref{alg:AddandFix2} computes an \eqx{} allocation in pseudo-polynomial time. \label{thm:Monotone}
\end{theorem}
\begin{proof}
We first show that the algorithm terminates in pseudo-polynomial time.  Claim~\ref{claim:AddandFixIteration} implies that the number of outer iterations is at most $\sum_{i \in N} v_i(M)$ $\le n V_{\max}$, where $V_{\max} := \max_i v_i(M)$.  Each outer iteration consists of an Add phase (in which at most $m$ goods are included in agent $p$'s bundle) and a Fix phase (in which at most $m$ goods are removed from agent $p$'s bundle). Each execution of these phases  requires at most $m$ calls to the value oracle to find the required goods $g^*$ and $\widehat{g}$. Hence, it follows that the algorithm terminates in $O(m^2 n V_{\max})$ time.

Next, we show that the allocation computed by the algorithm is indeed \eqx{}; our proof is via induction on the number of outer iterations. Initially, the allocation is empty, which is trivially \eqx{}. 

For the inductive step, fix any outer iteration and let $p$ be the poorest agent selected in that iteration. Write $\alloc=(A_1, \ldots, A_n)$ for the allocation at the beginning of the outer iteration and $\mathcal{B}=(B_1, \ldots, B_n)$ for the allocation obtained after the outer iteration. Note that $B_i = A_i$ for all agents $i \neq p$. This observation and the induction hypothesis imply that any \eqx{} violation must involve agent $p$. Further, Claim~\ref{claim:AddandFixIteration} gives us $v_p(B_p) \ge v_p(A_p)$.

To show that allocation $\mathcal{B}$ is \eqx{}, we need to show that for any agent $i \in N$, the removal of any good $g \in B_i$ makes $i$ a poorest agent. This condition holds---via the induction hypothesis---for all agents $i \neq p$; recall that $B_i = A_i$ for all $i \neq p$, and $v_p(B_p) \ge v_p(A_p)$. 

For agent $p$, note that after the completion of the Fix phase, the removal of any good from $p$'s bundle reduces its value to at most $v_{p'}(A_{p'}) = v_{p'}(B_{p'})$. Since $p$ was the poorest and $p'$ was the second poorest agent in allocation $\alloc$, this implies that the removal of any good from $B_p$ would make agent $p$ the poorest agent in $\mathcal{B}$ as well: $v_p(B_p \setminus \{ g\}) \leq v_j(B_j)$ for all $j \in N$ and each good $g \in B_p$. 

The theorem stands proved. 
\end{proof}

\subsection{Weakly Well-Layered Valuations}
\label{section:MonotoneWWL}
The following theorem asserts that for weakly well-layered valuations, Algorithm~\ref{alg:AddandFix2} computes an \eqx{} allocation in polynomial time.
\begin{theorem}
\label{thm:WWLPolyTime}
Given any fair division instance in which all the agents have monotone, weakly well-layered valuations, Algorithm~\ref{alg:AddandFix2} computes an \eqx{} allocation in polynomial time.
\end{theorem}
\begin{proof}
The monotonicity of agents's valuations ensures that the allocation returned by Algorithm~\ref{alg:AddandFix2} is \eqx{} (Theorem \ref{thm:Monotone}). Hence, it remains to prove that, under weakly well-layered valuations, the algorithm terminates in polynomial time. Towards this, we will show that, in fact, the Fix phase never executes when all the valuations are weakly well-layered. Hence, in every outer iteration of Algorithm \ref{alg:AddandFix2} the number of unassigned goods strictly decreases, and the algorithm terminates in polynomial time. 

We will show that the Fix phase never executes via an inductive argument. In the base case---i.e., in the very first outer iteration---we have $A_i = \emptyset$ for all agents $i$. Now, for the first iteration, write $S$ to denote the subset of goods assigned to the selected agent $p$ in the Add phase. Further, let $g^*$ denote the last good assigned in the Add phase. Then, by the loop-execution condition in Line~\ref{line:AddCondition}, we have $v_p(S \setminus \{ g^* \}) = 0$, since all other agents have value $0$. Further, given that $v_p$ is weakly well-layered and the set $S \setminus \{ g^* \}$ is populated greedily, any set of goods of cardinality $|S|-1$ has value $0$. Hence, upon the removal of any good $g$ from $S$, agent $p$ has value $0$ for the remaining subset, since it has size $|S|-1$. %, i.e., $v_p(S \setminus \{g \}) = 0$ for all goods $g \in S$. 
Therefore, the Fix phase (see Line~\ref{line:FixCondition}) will not execute in the first outer iteration. %That is, the set of unassigned items, $U$, monotonically decreases. 

For the inductive step, fix an outer iteration. Let $\alloc=(A_1, \ldots, A_n)$ be the allocation at the beginning of the iteration and $\overline{U} = M \setminus \left(\cup_i A_i \right)$ be the set of unassigned goods. For the poorest agent $p$ selected in the iteration, consider the bundle $A_p$ and write $\overline{S} \subseteq \overline{U}$ to denote the subset of goods assigned to $p$ in the Add phase of the iteration. In addition, let $\overline{g} \in \overline{S}$ be the last good assigned in the Add phase. %With $r := |A_p \cup \overline{S}|$, 
The following Claim asserts that all strict %size-$(r-1)$ 
subsets $T \subsetneq \left(A_p \cup \overline{S}\right)$ have value $v_p(T) \leq v_p \left( \left(A_p \cup \overline{S} \right)\setminus \{ \overline{g} \} \right)$. 

\begin{restatable}{claim}{ClaimAddFixWWL}
\label{claim:AddandFixWWL}
For weakly well-layered valuation $v_p$ we have 
\begin{align*}
\left(A_p \cup \overline{S} \right)\setminus \{ \overline{g} \}  \ \in \argmax_{X \subsetneq A_p \cup \overline{S}} v_p(X).
\end{align*}
\end{restatable}

We use Claim \ref{claim:AddandFixWWL} to complete the inductive step; the proof of the claim itself appears at the end of the subsection. The execution criterion of the Add phase (Line \ref{line:AddCondition}) implies that before good $\overline{g}$ was included in the agent $p$'s bundle its value was at most $v_{p'}(A_{p'})$, i.e., $v_p\left( \left(A_p \cup \overline{S} \right) \setminus \{ g \} \right) \le v_{p'}(A_{p'})$. Hence, via Claim \ref{claim:AddandFixWWL}, for every strict subset $T \subsetneq A_p \cup \overline{S}$ we have $v_p(T) \leq v_{p'}(A_{p'})$. The execution condition for the Fix phase will hence not be satisfied. This completes the proof of the theorem. 
\end{proof}

\begin{proof}[Proof of Claim~\ref{claim:AddandFixWWL}]
The induction hypothesis implies that Fix phase has not executed so far in the algorithm. Hence, during the previous iterations, the set of unassigned goods has decreased monotonically and the agents bundles, $A_i$s, have increased monotonically. We can, hence, index the goods in $A_p$ in the order in which they were included in the bundle, i.e., $A_p =\{g_1, g_2, \ldots, g_\ell\}$, where $g_1$ was the first good included in the bundle, and $g_\ell$ is the most recent. 
For each index $t \leq \ell$, write the prefix subset $G_t := \{g_1, \ldots, g_t\}$ and $G_0 := \emptyset$. Also, let $U_t$ denote the set of unassigned goods from which good $g_t$ was selected by agent $p$. As observed previously, we have $U_1 \supseteq U_2 \supseteq \ldots \supseteq U_\ell \supseteq \overline{U}$.

Therefore, for each index $t <\ell$, the following containment holds: $U_t \supseteq \overline{U} \cup \{g_t, g_{t+1}, \ldots, g_\ell\}  = \overline{U} \cup \left( G_\ell \setminus G_{t-1} \right)$.  Note that good $g_t$ was selected from $U_t$ by agent $p$ greedily to maximize the marginal contribution over $G_{t-1}$. That is, i.e., ${\displaystyle g_t \in \argmax_{g \in U_t} \left( v_p(G_{t-1} \cup g) - v_p(G_{t-1}) \right)}$ and, hence, the above-mentioned containment gives us 
\begin{align}
g_t \in \argmax_{g \in \overline{U} \cup \left( G_\ell \setminus G_{t-1} \right)} \left( v_p(G_{t-1} \cup g) - v_p(G_{t-1}) \right) \label{eqn:max-marg}
\end{align}
By definition, we have $G_\ell = A_p$. Hence, equation (\ref{eqn:max-marg}) implies that, in the set $\left(\overline{U} \cup A_p \right) \setminus G_{t-1}$, the good $g_t$ has maximum marginal contribution over $G_{t-1}$:
\begin{align}
g_t \in \argmax_{g \in \left(\overline{U} \cup A_p \right) \setminus G_{t-1}} \left( v_p(G_{t-1} \cup g) - v_p(G_{t-1}) \right) \label{eqn:max-marg1}
\end{align} 
Equation (\ref{eqn:max-marg1}) holds for all indices $t \leq \ell$. Therefore, the greedy algorithm (as specified in Definition \ref{definition:well-layered} and when executed for valuation $v_p$) would lead to subset $G_t=\{g_1, \ldots, g_t\}$, for each index $t \leq \ell$. Furthermore, with $G_\ell = A_p$, and given the selection criterion of the subset $\overline{S} \subseteq \overline{U}$ in the Add phase, the greedy algorithm when executed for $\left( |A_p \cup \overline{S}| -1 \right)$ steps over the goods in $A_p \cup \overline{U}$ would lead to the subset $(A_p \cup \overline{S}) \setminus \{ \overline{g}\}$; recall that $\overline{g} \in \overline{S}$ was the last good included in $p$'s bundle in the Add phase.   

Hence, given that $v_p$ is a weakly well-layered valuation (Definition \ref{definition:well-layered}), we obtain 
\begin{align*}
\displaystyle \left(A_p \cup \overline{S} \right)\setminus \{ \overline{g} \} \in \argmax_{X \subsetneq A_p \cup \overline{S}} \ v_p(X).
\end{align*}
The claim stands proved. 
\end{proof}

\subsection{Approximate \eqx{} Allocations}
\label{section:monotone-apx}
As mentioned previously, for general monotone valuations, computing an \eqx{} allocation is a {\rm PLS}-hard problem. Complementing this hardness result, this section establishes that an approximately \eqx{} allocation can be computed efficiently. In particular, for parameter $\varepsilon \in [0,1]$, an allocation $\alloc$ is said to be an $(1 - \varepsilon)$-\eqx{} allocation if for every pair of agents $i, j \in N$ and for each good $g \in A_i$ we have $(1- \varepsilon) \ v_i(A_i \setminus \{ g \}) \leq v_j(A_j)$. Hence, in an $(1- \varepsilon)$-\eqx{} allocation, removing any good from any agent $i$'s bundle brings down $i$'s value to below $\frac{1}{1- \varepsilon}$ times the minimum. Also, note that $\varepsilon = 0$ corresponds to an exact \eqx{} allocation. 
%Lastly, we show that even though for general monotone valuations, obtaining an \eqx{} allocation is PLS-complete, an approximate \eqx{} allocation can be obtained efficiently. For this result, given $\alpha \ge 1$, an allocation $A$ is said to be an $\alpha$-\eqx{} allocation if 

We modify Algorithm~\ref{alg:AddandFix2} to obtain an approximately \eqx{} allocation. We replace Lines~\ref{line:AddCondition} and~\ref{line:FixCondition} in Algorithm~\ref{alg:AddandFix2} with their approximate versions as follows:\footnote{Recall that Lines~\ref{line:AddCondition} and~\ref{line:FixCondition} in Algorithm~\ref{alg:AddandFix2} check the execution condition for the Add and Fix phases, respectively.}

\noindent \ref{line:AddCondition}: \ \textbf{while} $( 1- \varepsilon) \ v_p(A_p) \leq v_{p'}(A_{p'})$ \textbf{and} $U \neq \emptyset$ \textbf{do}...

\noindent \ref{line:FixCondition}: \ \textbf{while} there exists $\widehat{g} \in A_p$ such that $( 1- \varepsilon) v_p(A_p \setminus \{ \widehat{g} \}) > v_{p'}(A_{p'})$ \textbf{do}...

The following theorem provides our main approximation guarantee under monotone valuations. The proof of this result is similar to that of Theorem~\ref{thm:Monotone}; see Appendix \ref{appendix:MonApx}.   

\begin{restatable}{theorem}{TheoremMonApx}
\label{thm:MonotoneApproximation}
Given parameter $\varepsilon \in (0,1)$ and any fair division instance with monotone valuations, a $(1-\varepsilon)$-\eqx{} allocation can be computed in $O \left(\frac{m^2n}{\varepsilon} \log V_{\max} \right)$ time. %, where $V_{\max} := \max_i \ v_i(M)$.
\end{restatable}

\begin{remark}
Note that the results obtained in this section for monotone nondecreasing valuations, also hold for monotone nonincreasing valuations (i.e., chores, rather than goods), with appropriate modifications to the algorithms.
\end{remark}

\section{Additive Nonmonotone Valuations}  
\label{sec:Chores}

We now present our results for nonmonotone valuations. We will focus primarily on additive valuations, and will show that even in this case, \eqx{} allocations may not exist. Furthermore, \eqx{} allocations can be hard to compute even if they do exist. Some of our results extend beyond additive valuations, however, in this section, we will conform to additive valuations for ease of exposition. 

It is known that an \eqx{} allocation can be computed in polynomial time if all the agents have identical, additive valuations~\cite{AzizR20}. Complementing this positive result, we next show that an \eqx{} allocation may not exist among agents that have nonidentical (and nonmonotone) valuations. Our negative result holds even for two agents with additive, normalized valuations. The fair division instance demonstrating this nonexistence of \eqx{} allocations is given in Table~\ref{table:NormalisedExample}. In particular, the instance  highlights that if there are items with positive value for one agent and negative value for another---i.e., the valuations are \emph{subjective} rather than {objective}---then an \eqx{} allocation may not exist. Further, in such instances, it is {\rm NP}-hard to determine if an \eqx{} allocation exists.

\begin{table}[!ht]
\centering
    \begin{tabular}{|c|c|c|c|}
    \hline
     & $x_1$ & $x_2$ & $x_3$ \\ \hline
     Agent 1 & $+1$ & $-1$ & +100 \\ \hline
     Agent 2 & $-1$ & $+1$ & +100 \\ \hline     
    \end{tabular}
    \caption{Example showing nonexistence of \eqx{} allocations in Theorem \ref{theorem:no-mas-eqx}.} 
     \label{table:NormalisedExample}
\end{table}

An example for nonexistence is given in Table~\ref{table:NormalisedExample}. The computational hardness is established via a reduction from the \Part \ problem in the following theorem.  

\begin{restatable}{theorem}{TheoremNonExistence}
\label{theorem:no-mas-eqx}
An \eqx{} allocation may not exist for two agents with nonmonotone, additive, normalised valuations. Further, in such instances, it is weakly {\rm NP}-hard to determine whether an \eqx{} allocation exists or not.
\end{restatable}
\begin{proof}
We show nonexistence of \eqx{} allocations in the instance given in Table~\ref{table:NormalisedExample}. Assume, towards a contradiction, that there is an \eqx{} allocation $\alloc= (A_1, A_2)$ in the specified instance. Note that, here, under \eqx{}, for $i =1$ and $i=2$, agent $i$ cannot receive two items that are both goods for $i$. Otherwise, removing the good, besides $x_3$, from $i$'s bundle would not drive $v_i(A_i)$ below the other agent's value. That is, for $\alloc$ to be \eqx{}, it must hold that $|A_1 \cap \{x_1, x_3\}| \leq 1$ and $|A_2 \cap \{x_2, x_3\}| \leq 1$.

Therefore, if $x_3 \in A_1$, then $x_1 \in A_2$. Now, note that $x_1$ is a chore for agent 2 and even after removing it from $A_2$, we have $v_2(A_2 \setminus \{x_1\}) \leq 1$. In addition, the containment $x_3 \in A_1$ ensures $v_1(A_1) \geq 99$. These inequalities imply that even after the removal of item $x_1$ (which is a chore for agent 2) from $A_2$, agent 2's value does not exceed $v_1(A_1)$. That is, the \eqx{} condition cannot hold if $x_3 \in A_1$. 

On the other hand, if $x_3 \in A_2$, then we have $x_2 \in A_1$. Item $x_2$ is a chore for agent 1, and in the current case the following inequalities hold: $v_1(A_1 \setminus \{x_2\}) \leq 1$ and $v_2(A_2) \geq 99$. Therefore, again, allocation $\alloc$ does not uphold the \eqx{} criterion. Overall, we obtain that the instance does not admit an \eqx{} allocation. 

To establish \NP-hardness, we provide a reduction from \Part, known to be \NP-complete. An instance of \Part{} consists of a set of $M$ positive integers $a_1$, $\ldots$, $a_M$. The problem is to determine if there exists a subset $S \subseteq [M]$ such that $\sum_{i \in S} a_i$ $= \sum_{i \not \in S} a_i$. 

Given an instance of \Part, we construct a fair division instance as follows. There are two agents $1$ and $2$, and $M+2$ items. The first two items, $x_1$ and $x_2$, are similar to the first two items in Table~\ref{table:NormalisedExample}. That is, item $x_1$ has value $+1$ for agent $1$ and $-1$ for agent $2$. Item $x_2$ has value $-1$ for agent $1$ and $+1$ for agent $2$. The remaining items, $g_1$, $\ldots$, $g_M$, are all identical goods, with item $g_i$ possessing value $2a_i$ for both agents (i.e., twice the value of corresponding integer in the \Part{} instance).

We show that that the \Part{} instance has a solution iff the constructed fair division instance admits an \eqx{} allocation. \\

\noindent
{\it Forward Direction:} Suppose there exists $S \subseteq [n]$ such that $\sum_{i \in S} a_i = \sum_{i \not \in S} a_i$. Then, consider the allocation that assigns agent $1$ items $x_1$, $x_2$, and all goods with indices in $S$. This results in an equitable (hence, an \eqx{}) allocation in which the both the agents' values are equal to $\sum_{i \in M} a_i$. \\

\noindent
{\it Reverse Direction:} In any allocation $\alloc = (A_1, A_2)$ of the fair division instance, and no matter how $x_1$ and $x_2$ are allocated, the absolute difference in value $|v_1(A_1) - v_2(A_2)|$ must be even. For the reverse direction of the reduction, assume that the constructed fair division instance admits an \eqx{} allocation  $\alloc =(A_1, A_2)$. Using a case analysis based on the assignment of the items $x_1$ and $x_2$ between the two agents, one can show that under the \eqx{} allocation $\alloc$ we must have $|v_1(A_1) - v_2(A_2)| \leq 1$. These bounds on the value difference imply that, in the \eqx{} allocation, the agents' values are in fact equal
\begin{align}
v_1(A_1) = v_2(A_2) \label{eq:val-equality}
\end{align}   

Further, note that there are three possible cases when considering the assignments of the items $x_1$ and $x_2$ between the two bundles $A_1$ and $A_2$: 
\begin{itemize}
\item[(i)] Both the items are assigned to the same agent. Since $v_1(\{x_1, x_2\}) = v_2(\{x_1, x_2\}) = 0$, in this case we have $v_1(A_1) = v_1(A_1 \setminus \{x_1, x_2\})$ and $v_2(A_2) = v_2(A_2 \setminus \{x_1, x_2\})$. Using equation (\ref{eq:val-equality}), we further obtain 
\begin{align}
v_1(A_1 \setminus \{x_1, x_2\}) = v_2(A_2 \setminus \{x_1, x_2\}) \label{eq:eqx-part}
\end{align}
\item[(ii)] The items $x_1$, $x_2$ are split between the bundles with $x_1 \in A_1$ and $x_2 \in A_2$. In this case, $v_1(A_1) = v_1(A_1 \setminus \{x_1\}) + 1$ and $v_2(A_2) = v_2(A_2 \setminus \{x_2\}) + 1$. Again, the equality $v_1(A_1 \setminus \{x_1, x_2\}) = v_2(A_2 \setminus \{x_1, x_2\})$ holds.
\item[(iii)] The items $x_1$, $x_2$ are split between the bundles with $x_1 \in A_2$ and $x_2 \in A_1$. Here, $v_1(A_1) = v_1(A_1 \setminus \{x_2\}) - 1$ and $v_2(A_2) = v_2(A_2 \setminus \{x_1\}) - 1$. This leads to the same equality obtained in the previous cases: $v_1(A_1 \setminus \{x_1, x_2\}) = v_2(A_2 \setminus \{x_1, x_2\})$.
\end{itemize}
That is, equality (\ref{eq:eqx-part}) holds in all three cases. Hence, the goods in the set $A_1 \setminus \{x_1, x_2\}$ correspond to elements $S \subseteq [M]$ with the property that 
\begin{align*}
\sum_{i \in S} a_i & = \frac{1}{2} v_1(A_1 \setminus \{x_1, x_2\} ) \\
& = \frac{1}{2} v_2(A_2 \setminus \{x_1, x_2\}) \tag{via (\ref{eq:eqx-part})}\\
& =  \sum_{j \notin S} a_j.
\end{align*}
Therefore, the set $S$ is a solution for \Part{} instance. This establishes the reverse direction of the reduction and completes the proof of the theorem. 
\end{proof}

\subsection{Pseudo-polynomial Time Algorithm for a Fixed Number of Agents with Additive Valuations}
\label{sec:PseudoPoly}
We now introduce a dynamic programming algorithm that operates in pseudo-polynomial time and can determine the existence of an \eqx{} allocation in a given fair division instance with a constant number of agents with additive subjective valuations. For ease of exposition, we will limit our discussion to the case of two agents. The algorithm and analysis can be easily extended to a constant number of agents.

Let us assume that the items in $M$ are labelled $x_1,x_2,\hdots,x_m$. Define $H \coloneqq \max_{i \in N}\max_{x_j\in M} |v_i(x_j)|$. The dynamic programming algorithm has two components: table-filling and table-checking. The table-filling component of algorithm maintains a binary-valued table $A[j,w_1,w_2,h_1,d_1,h_2,d_2]$. An entry \\$A[j,w_1,w_2,h_1,d_1,h_2,d_2]=1$ if there exists an allocation of the first $j$ items such that agent 1's value for its bundle is $w_1$, agent $2$'s value for its bundle is $w_2$, the value of minimum valued good in agent $1$'s bundle is $h_1$ while that of agent $2$ is $h_2$ and the value of maximum valued chore in agent $1$'s bundle is $d_1$ while that of agent $2$ is $d_2$. For an agent $i$, we thus have $w_i \in [-mH,mH],\;h_i \in [0,H],\;d_i \in [-H,0]$ and the parameter $H$ is as defined before. Note that for values of the parameters outside the range, we have $A[j,w_1,w_2,h_1,d_1,h_2,d_2]=0$.  

Once we have values for all the entries of the table, checking the existence of an \eqx{} allocation is a simple task -- it is carried out by the table-checking component of algorithm. This component examines $A[m,w_1,w_2,h_1,d_1,h_2,d_2]$ for all values of $w_i,h_i$ and $d_i$ for $i \in \{1,2\}$. For each table entry with $A[m,w_1,w_2,h_1,d_1,h_2,d_2]=1$, the component verifies the \eqx{} conditions: If $w_2 > w_1$ (agent 2 values its bundle higher than agent 1 values its own), then the we check whether $w_2-h_2 \leq w_1$ (removal of minimum valued good from agent 2's bundle causes inquity to vanish) and $w_1-d_1 \geq w_2$ (removal of maximum valued chore from agent 1's bundle causes inequity to vanish) are satisfied, or not. Similarly,  for the case $w_1 > w_2$, we check whether the inequalities $w_1-h_1 \leq w_2$ and $w_2-d_2 \geq w_1$ are satisfied, or not. If the algorithm finds a filled (with $1$) entry in the table for which the \eqx{} conditions hold, then it declares that an \eqx{} allocations exists. Otherwise, if no such entry is found, the algorithm declares that the given instance does not admit an \eqx{} allocation. 
 
\begin{algorithm}[!htp]
  \caption{Table-filling}
  \label{algo:TableFill}
  \textbf{Input:} Instance $(N,M,\mathcal{V})$ with two agents and additive subjective valuations. \\
  \textbf{Output:} Binary-valued table $A$
  \begin{algorithmic}[1]  

\STATE \begin{align*}
& f_1 =\begin{cases}
1 & \text{if }v_1(x_1)\geq 0,~ w_1=v_1(x_1),~ w_2=0,~ h_1=v_1(x_1),\\
& d_1=0,~ h_2=0,~ d_2=0 \qquad \qquad \text{\COMMENT{$x_1$ is a good for agent 1, and assigned to agent 1}}\\
& \qquad \qquad \qquad \qquad \qquad \qquad \text{OR}\\
& \text{if }v_1(x_1) < 0,~ w_1=v_1(x_1),~ w_2=0,~ h_1=0,\\
& d_1=v_1(x_1),~ h_2=0,~ d_2=0 \qquad \text{\COMMENT{$x_1$ is a chore for agent 1, and assigned to agent 1}}\\
0 & \text{Otherwise}
\end{cases}\\
&f_2=\begin{cases}
1 & \text{if }v_2(x_1)\geq 0,~ w_1=0,~ w_2=v_2(x_1),~ h_1=0,\\
& d_1=0,~ h_2=v_2(x_1),~ d_2=0\qquad \text{\COMMENT{$x_1$ is a good for agent 2, and assigned to agent 2}}\\
& \qquad \qquad \qquad \qquad \qquad \qquad \text{OR}\\
& \text{if }v_2(x_1) < 0,~ w_1=0,~ w_2=v_2(x_1),~ h_1=0,\\
& d_1=0,~ h_2=0,~ d_2=v_2(x_1) \qquad \text{\COMMENT{$x_1$ is a chore for agent 2, and assigned to agent 2}}\\
0 & \text{Otherwise}\\
\end{cases}\\
&\text{for}\; i \in \{1,2\},\;w_i \in [-mH,mH],\;h_i \in [0,H],\;d_i \in [-H,0]\\
&\qquad A[1,w_1,w_2,h_1,d_1,h_2,d_2] = f_1 \lor f_2
\end{align*}
\STATE for $j \in [2,m]$ \\
   \quad for $ i \in \{1,2\},\;w_i \in [-mH,mH],\;h_i \in [0,H],\;d_i \in [-H,0]$ 
\begin{align*}
&f_1 = \begin{cases}\lor_{h_1^{'} \in[h_1,H]}A[j-1,w_1-v_1(x_j),w_2,h_1^{'},d_1,h_2,d_2]  & \text{if } v_1(x_j) \geq 0, ~v_1(x_j)=h_1\\
A[j-1,w_1-v_1(x_j),w_2,h_1,d_1,h_2,d_2] & \text{if }v_1(x_j) \geq 0, ~v_1(x_j)>h_1\\
0 & \text{if }v_1(x_j) \geq 0, ~v_1(x_j)<h_1\\
\lor_{d_1^{'} \in [-H,d_1]}A[j-1,w_1-v_1(x_j),w_2,h_1,d_1^{'},h_2,d_2] & \text{if } v_1(x_j) < 0, ~v_1(x_j)=d_1\\
A[j-1,w_1-v_1(x_j),w_2,h_1,d_1,h_2,d_2] & \text{if }v_1(x_j) < 0,~v_1(x_j)<d_1\\
0 & \text{if }v_1(x_j) < 0, ~v_1(x_j)>d_1
\end{cases}\\
&f_2= \begin{cases}\lor_{h_2^{'} \in[h_2,H]}A[j-1,w_1,w_2-v_2(x_j),h_1,d_1,h_2^{'},d_2]  & \text{if } v_2(x_j) \geq 0, ~v_2(x_j)=h_2\\
A[j-1,w_1,w_2-v_2(x_j),h_1,d_1,h_2,d_2] & \text{if }v_2(x_j) \geq 0, ~v_2(x_j)>h_2\\
0 & \text{if }v_2(x_j) \geq 0, ~v_2(x_j)<h_2\\
\lor_{d_2^{'} \in [-H,d_2]}A[j-1,w_1,w_2-v_2(x_j),h_1,d_1,h_2,d_2^{'}] & \text{if } v_2(x_j) < 0, ~v_2(x_j)=d_2\\
A[j-1,w_1,w_2-v_2(x_j),h_1,d_1,h_2,d_2] & \text{if }v_2(x_j) < 0, ~v_2(x_j)<d_2\\
0 & \text{if }v_2(x_j) < 0, ~v_2(x_j)>d_2
\end{cases}\\
&A[j,w_1,w_2,h_1,d_1,h_2,d_2] = f_1 \lor f_2
\end{align*}
\RETURN $A$.
\end{algorithmic}
\end{algorithm}

\begin{algorithm}[!htp]
  \caption{Table-checking}
  \label{algo:TableCheck}
  \textbf{Input:} Instance $(N,M,\mathcal{V})$ with two agents and additive subjective valuations. \\
  \textbf{Output:} \eqx{} allocation exists and not.
  \begin{algorithmic}[1]  
\STATE $A=$Table-filling($(N,M,\mathcal{V})$)
\FORALL{all entries with $A[m,w_1,w_2,h_1,d_1,h_2,d_2]=1$}  
%\begin{ALC@g}
      \IF{$w_1 < w_2$}
          \IF{$w_2-h_2 \leq w_1$ and $w_1-d_1 \geq w_2$}
              \RETURN ``\eqx{} allocation exists''
          \ENDIF
     \ELSE
          \IF{$w_1-h_1 \leq w_2$ and $w_2-d_2 \geq w_1$}
              \RETURN ``\eqx{} allocation exists''
          \ENDIF
    \ENDIF
\ENDFOR
%\end{ALC@g}
\RETURN ``\eqx{} allocation does not exist''
\end{algorithmic}
\end{algorithm}

\begin{restatable}{lemma}{LemmaTablefill}
\label{lemma:LemmaTablefill}
Given any fair division instance involving two agents with additive subjective valuations, Algorithm \ref{algo:TableFill} correctly computes the table $A[j,w_1,w_2,h_1,d_1,h_2,d_2]$ for all values of the parameters $j,w_i,h_i $ and $d_i$ for $i \in \{1,2\}$.
\end{restatable}
\begin{proof}
We use induction on the number of items involved in an allocation.

Let us consider the base case. Here, we consider allocations involving only one item $x_1$. An allocation can have either item $x_1$ assigned to the agent $1$ or item $x_1$ assigned to the agent $2$. Consider the case in which item $x_1$ is assigned to the agent $1$. If item $x_1$ is a good for agent $1$ then we have $A[1,v_1(x_1),0,v_1(x_1),0,0,0]=1$. Otherwise, if $x_1$ is a chore for agent 1, we have $A[1,v_1(x_1),0,0,v_1(x_1),\\0,0]=1$. Similarly if item $x_1$ is assigned to the agent $2$, we get either $A[1,0,v_2(x_1),0,0,v_2(x_1),0]=1$ or  $A[1,0,v_2(x_1),0,0,0,v_2(x_1)]=1$ depending on whether $x_1$ is a good or a chore for agent 2. Since these are the only two feasible allocations of $x_1$, for other values of the parameters, the algorithm assigns 0 to $A[1,w_1,w_2,h_1,d_1,h_2,d_2]$. Thus, the algorithm computes $A[1,w_1,w_2,h_1,d_1,h_2,d_2]$ correctly for all values of the parameters and the base case is verified.

For the inductive case, let us assume that the algorithm correctly computes $A[j-1,w_1,w_2,h_1,d_1,h_2,d_2]$ for all values of the parameters. To compute $A[j,w_1,w_2,h_1,d_1,h_2,d_2]$, as in the base case, the algorithm considers four cases: $(i)$ the item $x_j$ is assigned to agent 1, and is a good for agent 1$(ii)$ the item $x_j$ is assigned to agent 1, and is a chore for agent 1 $(iii)$  the item $x_j$ is assigned to agent 2, and is a good for agent 2 and $(iv)$ the item $x_j$ is assigned to agent 2, and is a chore for agent 2. Consider case $(i)$, the other cases are symmetric. In case $(i)$, we have three possibilities: $(1)$ $v_1(x_j)=h_1$, i.e., the value of $x_j$ for  agent $1$ is equal to the desired value of minimum valued good in agent $1$'s bundle $(2) \;v_1(x_j) > h_1$ and $(3) \; v_1(x_j) < h_1$. 
\begin{itemize}
\item[(1)] If $v_1(x_j)=h_1$  then the algorithm checks if there exists an allocation of the first $j-1$ items, $w_1-v_1(x_j)$ as agent 1's value for its bundle and the desired value of minimum valued good in agent 1's bundle greater than or equal to  $h_1$ with values for all other parameters being same. By induction hypothesis, we know that the value of $A[j-1,w_1,w_2,h_1,d_1,h_2,d_2]$, for all values of the parameters are computed correctly. Thus, the algorithm correctly updates the value of $f_1$ in this case. \\
\item[(2)] Suppose $v_1(x_j) > h_1$, then the algorithm checks if there exists an allocation of first $j-1$ items and $w_1-v_1(x_j)$ as agent 1's value for its bundle with values for other parameters being same. Again, by induction hypothesis we can conclude that the algorithm updates $f_1$ correctly.\\
\item[(3)] In case $v_1(x_j) < h_1$, it's not possible to get an allocation satisfying the parameter $h_1$. Hence, the algorithm assigns $0$ to $f_1$. 
\end{itemize}
Thus, the algorithm updates $f_1$ correctly. Since the other cases are symmetric, we also conclude that the algorithm updates $f_2$ correctly. The variables $f_1$ and $f_2$ are then combined using $\lor$ to obtain the value of $A[j,w_1,w_2,h_1,d_1,h_2,d_2]$. Thus, the value of $A[j,w_1,w_2,h_1,d_1,h_2,d_2]$ is computed correctly by the algorithm.
\end{proof}

\begin{restatable}{theorem}{TheoremPseudoPoly}
\label{theorem:PseudoPoly}
Given any fair division instance involving two agents with additive subjective valuations, Algorithm \ref{algo:TableCheck} correctly determines the existence of an \eqx{} allocation in pseudo-polynomial time.
\end{restatable}
\begin{proof}
In the first step, Algorithm \ref{algo:TableCheck} invokes table-filling algorithm which returns a binary-valued table $A$. We saw in Lemma \ref{lemma:LemmaTablefill} that this table is computed correctly by the table-filling algorithm. 

Algorithm \ref{algo:TableCheck} then verifies \eqx{} conditions for all entries $A[m,w_1,w_2,h_1,d_1,h_2,d_2]=1$, i.e., if $w_2 > w_1$ (agent 2 values its bundle higher than agent 1 values its own) then it checks whether $w_2-h_2 \leq w_1$ (removal of minimum valued good from agent 2's bundle causes inquity to vanish) and $w_1-d_1 \geq w_2$ (removal of maximum valued chore from agent 1's bundle causes inquity to vanish) are satisfied or not and similarly for the case $w_1 > w_2$, it checks whether $w_1-h_1 \leq w_2$ and  $w_2-d_2 \geq w_1$ are satisfied or not. The algorithm returns 1 if \eqx{} conditions are satisified for some $w_i,h_i$ and $d_i$ for $i \in \{1,2\}$ and returns 0 otherwise.
This completes the proof of correctness.

We now show that the Algorithm \ref{algo:TableCheck} terminates in pseudo-polynomial time. The algorithm first invokes table-filling algorithm which takes $O(m^{3}H^{7})$ time to complete its execution. It then verifies \eqx{} conditions and decides the existence of \eqx{} allocation in the given instance. This part of the execution takes $O(m^{2}H^{6})$ time. Thus, the overall execution time of the algorithm is $O(m^{3}H^{7})$. 
\end{proof}
Note that if there are $c>2$ agents, the table looks like $A[j,w_1,w_2,\hdots,w_c,h_1,d_1,h_2,d_2,\hdots,h_c,d_c]$ and the runtime will be $O(m^{c+1}H^{3c+1}+m^{c}H^{3c})$, where $O(m^{c+1}H^{3c+1})$ and $O(m^{c}H^{3c})$ are the running times of table-filling and table-checking algorithms respectively. 

\subsection{Strong \NP-hardness for Arbitrary Agents with Additive Valuations}
\label{sec:StrongNPHardness}
We now show that if the number of agents is part of the input, determining whether there exists an \eqx{} allocation is strongly \NP-hard. The result is established using a reduction from 3-\Part{} which is known to be strongly \NP-hard. 

\begin{restatable}{theorem}{TheoremStrongNP-hardness}
\label{theorem:strongNP-hard}
Given any fair division instance with additive subjective valuations, determining whether there exists an \eqx{} allocation is strongly \NP-hard.
\end{restatable}
\begin{proof}
The proof proceeds via a reduction from 3-\Part{}, which is known to be strongly \NP-hard. In 3-\Part{}, we are given a multiset of $3n$ positive integers $S=\{a_1,a_2,\hdots,a_{3n}\}$ where $n \in \mathbb{N}$, the sum of all integers in $S$ is $nT$, i.e., $\sum_{i = 1}^{3n} a_i=nT$ and each integer $a_i$ satisfies $a_i \in \big(\frac{T}{4},\frac{T}{2}\big)$. The objective is to find a partition of $S$ into $n$ subsets $S_1,S_2,\hdots,S_n$ such that the sum of integers in each subset is $T$. Note that the condition $a_i \in \left(\frac{T}{4},\frac{T}{2}\right)$ ensures that each subset $S_i$ is of cardinality $3$. 

Given an instance of 3-\Part{}, we construct a fair division instance as follows. There are $n+1$ agents $1,2,\hdots,n+1$ and $3n+2$ items $x_1,x_2,\hdots,x_{3n+2}$. Agents $1,\hdots,n$ are called ``original" agents, and items $x_1,\hdots,x_{3n}$ are called ``original" items. Also, let $\Delta \gg 10nT$ be a sufficiently large positive integer.

Now, we define the valuations for the agents. We have $v_i(x_j)=a_j$ for each $i \in [n]$ and $j \in [3n]$. Hence, the items $x_1,\ldots, x_{3n}$ are identically valued by the agents $1,\hdots,n$ and their values correspond to the integers in $S$. Items $x_{3n+1}$ and $x_{3n+2}$ are also identically valued by the agents in $[n]$, with values $\frac{\Delta}{10}$ and $-\frac{\Delta}{10}$, respectively: $v_i(x_{3n+1})=\frac{\Delta}{10}$ and $v_i(x_{3n+2})=-\frac{\Delta}{10}$, for each $i \in [n]$. 

The agent $n+1$ values the items $x_1,\hdots,x_{3n}$ at $-\Delta$, i.e.,  $v_{n+1}(x_j)=-\Delta$ for each $j \in [3n]$. Further, agent $n+1$ values the remaining items $x_{3n+1}$ and $x_{3n+2}$ at $T$ and $1$, respectively. The valuations for the constructed fair division instance are depicted in Table \ref{table:InstanceStrongNP}.

\begin{table}[!ht]
\centering
\renewcommand{\arraystretch}{1.5}
    \begin{tabular}{| c | c | c | c | c | c | c |}
    \hline
     & $x_1$ & $x_2$ & $\cdot \cdot \cdot$ & $x_{3n}$ & $x_{3n+1}$ & $x_{3n+2}$\\ \hline
     Agent 1 & $a_1$ & $a_2$ & $\cdot \cdot \cdot$ & $a_{3n}$ & $\frac{\Delta}{10}$ & $-\frac{\Delta}{10}$ \\ \hline
     Agent 2 & $a_1$ & $a_2$ & $\cdot \cdot \cdot$ & $a_{3n}$ & $\frac{\Delta}{10}$ & $-\frac{\Delta}{10}$ \\ \hline
 $\cdot$ & $\cdot$ & $\cdot$ & $\cdot \cdot \cdot$ & $\cdot$ & $\cdot$ & $\cdot$ \\ \hline
$\cdot$ & $\cdot$ & $\cdot$ & $\cdot \cdot \cdot$ & $\cdot$ & $\cdot$ & $\cdot$ \\ \hline    
   Agent $n$ & $a_1$ & $a_2$ & $\cdot \cdot \cdot$ & $a_{3n}$ & $\frac{\Delta}{10}$ & $-\frac{\Delta}{10}$ \\ \hline 
Agent $n+1$ & $-\Delta$ & $-\Delta$ & $\cdot \cdot \cdot$ & $-\Delta$ & $T$ & $1$ \\ \hline
    \end{tabular}
    \caption{Fair division instance for establishing strong \NP-hardness} 
     \label{table:InstanceStrongNP}
\end{table}

We show that the 3-\Part{} instance has a solution iff the constructed fair division instance admits an \eqx{} allocation.\\

\noindent
{\it Forward Direction:} Suppose $S_1,S_2,\hdots,S_n$ is a solution of the 3-\Part{} instance. Then an \eqx{} allocation $A=(A_1,\hdots,A_{n+1})$ can be constructed as follows. If integer $a_j$ belongs to the subset $S_i$ then add item $x_j$ to the agent $i$'s bundle $A_i$. Assign items $x_{3n+1}$ and $x_{3n+2}$ to the agent $n+1$. In the resulting allocation, agents in $[n]$ have equal value $T$ for their bundles while agent $n+1$ values its bundle at $T+1$. This is clearly an \eqx{} allocation.\\

\noindent
{\it Reverse Direction:} Let $\alloc=(A_1, \ldots, A_n)$ be an \eqx{} allocation. We first claim that no original item is assigned to agent $n+1$. Assume, towards a contradiction, that some original item $x_j$ is assigned to the last agent $n+1$. Since $v_{n+1}(x_j)=-\Delta$, we have $v_{n+1}(A_{n+1}) < -\frac{\Delta}{2}$. In fact, if along with $x_j$, another original item, say $x_k$, is assigned to agent $n+1$, then even after the removal of $x_k$ we would have $v_{n+1}(A_{n+1} \setminus \{x_k \})  < -\frac{\Delta}{2}$. On the other hand, for each original agent $i \in [n]$, it must hold that $v_i(A_i) \geq - \frac{\Delta}{10}$. Hence, if $A_{n+1}$ contains two (or more) original goods, then the \eqx{} condition cannot hold for agent $n+1$. This observation implies that, in allocation $\alloc$, all the original goods, besides $x_j$, must have been assigned among the original agents. Therefore, there exists an agent $\hat{i} \in [n]$ whose bundle $A_{\hat{i}}$ is of size at least two. Further, considering the original agents' valuations, we obtain that $A_{\hat{i}}$ must contain a good $x$ and $v_{\hat{i}} (A_{\hat{i}} \setminus \{x \}) \geq - \frac{\Delta}{10} $. However, this inequality and the bound $v_{n+1}(A_{n+1}) < -\frac{\Delta}{2}$ contradict the fact that $\alloc$ is an \eqx{} allocation. Therefore, by way of contradiction, we have that no original item is assigned to agent $n+1$. 

We next show that both items $x_{3n+1}$ and $x_{3n+2}$ must be assigned to agent $n+1$.

\begin{itemize}
\item[(\rm i)] Assume, towards a contradiction, that $x_{3n+2} \notin A_{n+1}$. Since the original agents are symmetric, we can, without loss of generality, further assume that    $x_{3n+2} \in A_1$. We consider two cases, depending on if $x_{3n+1} \in A_1$ or not. In the first case (i.e., if $x_{3n+1} \in A_1$), we have $A_{n+1} = \emptyset$ and, hence, $v_{n+1}(A_{n+1}) = 0$. Here, for \eqx{} to hold, it must hold that removing a good from any other agent $\{1, \ldots, n\}$ should reduce its value to be at most zero. In particular, each of the original agents can get at most one original item. However, such a containment requirement cannot hold, since there are $3n$ original items (and agent $n+1$ cannot get any original items). Therefore, we obtain a contradiction, when both $x_{3n+2} \in A_1$ and $x_{3n+1} \in A_1$. 

In remains to address the case wherein $x_{3n+2} \in A_1$ and $x_{3n+1} \notin A_1$. Here, $v_1(A_1) < 0$, and, for \eqx{} to hold, it must hold that removing a good from any remaining agent $i \in \{2, \ldots, n\}$ should leave agent $i$ with a negatively valued bundle. {Note that this requirement does not hold even if $A_i = \emptyset$, since $v_i(A_i) = 0$ and $v_1(A_1)<0$.} Moreover, for the original agents every item, except $x_{3n+2}$, is a good and agent $n+1$ cannot be assigned any original item. Since such an allocation (of the $3n$ original items among $n$ original agents) is not possible, by way of contradiction, we get that $x_{3n+2} \notin A_{n+1}$.

\item[(\rm ii)] Next, assume, towards a contradiction, that $x_{3n+1} \notin A_{3n+1}$. Then, $A_{n+1}=\{x_{3n+2}\}$ and we have $v_{n+1}(A_{n+1})=1$. Further, items $x_1,\hdots,x_{3n+1}$ must be allocated to the $n$ original agents, who have value at least $1$ for each of the items. Hence, in such a case, some original agent must receive at least $4$ of these items. Removing an item still leaves her with a bundle of value at least $3$ and, hence, $\alloc$ cannot be \eqx{}. Therefore, we get that $x_{3n+1} \in A_{3n+1}$.
\end{itemize}

\noindent
The analysis above implies that, in \eqx{} allocation $\alloc$, items $x_1,\hdots,x_{3n}$ must be assigned to the agents in $[n]$, while items $x_{3n+1}$ and $x_{3n+2}$ must be assigned to the agent $n+1$. 

Furthermore, each agent in $[n]$ must possess a bundle of value at least $T$, since $v_{n+1}\left(A_{n+1} \setminus \{ x_{3n+2} \} \right) =  v_{n+1}(x_{3n+1})=T$. Since $\sum_{i = 1}^{3n} a_i=nT$, each agent must get value exactly $T$ from her bundle. Hence, an \eqx{} allocation induces a solution of the 3-\Part{} instance.

\end{proof}

Given that under subjective valuations \eqx{} allocations are not guaranteed to exist (Theorem \ref{theorem:no-mas-eqx}), we will focus on objective additive valuations in the remainder of this section. We will use $C$ to denote the set of chores and $G$ to denote the set of goods. Since valuations are objective, $M = G \cup C$. Further, given an allocation $\alloc$, we say agent $i$ has a goods violation if, for some good $g \in A_i$, we have $v_i(A_i \setminus \{g\}) > \min_k v_k(A_k)$. Similarly, agent $i$ has a chores violation if $v_i(A_i \setminus \{c\}) < \max_k v_k(A_k)$  for some chore $c \in A_i$.

\subsection{Two Agents with Additive Objective Valuations}
\label{sec:TwoAgents}

The section establishes that Algorithm \ref{algo:TwoAgents} efficiently computes an \eqx{} allocation between two (nonidentical) agents. A key property utilized in the analysis is that, if , in any iteration of Algorithm \ref{algo:TwoAgents}, item $g^*$ is assigned to agent $p$, then any chore $c$ assigned previously (i.e., in an earlier iteration) to the other agent $r$ must have greater absolute value. Similarly, if $c^*$ is assigned to agent $r$ in the considered iteration, any good $g$ assigned previously to the other agent $p$ must have greater absolute value. Using these bounds we show that the \eqx{} criterion is maintained as an invariant as the algorithm iterates.  

\begin{algorithm}[!htp]
  \caption{Two-Way Greedy Algorithm}
  \label{algo:TwoAgents}
  \textbf{Input:} Instance $(N,M,\mathcal{V})$ with two agents and additive objective valuations. \\
  \textbf{Output:} \eqx{} allocation $\alloc$.
  \begin{algorithmic}[1]  
  \STATE Initialize $\alloc = (\emptyset, \ldots, \emptyset)$ and item set $U :=  M$.
  \WHILE{$U \neq \emptyset$}
     \STATE Write ${\displaystyle r \in \arg \max_{i \in N} v_i(A_i)}$ and let $p$ be the other agent ($p \neq r$). 
      \STATE Set item ${\displaystyle g^*  \in \arg \max_{g \in U \cap G} v_p(g)}$ and item ${\displaystyle c^*  \in \arg \min_{c \in U \cap C} v_r(c)}$.
     \COMMENT{$g^*$ is most valuable good for $p$ and $c^*$ is least valuable chore for $r$.}
     \IF{$|v_p(g^{*})| > \vert v_r(c^{*})\vert$} \label{line:abs-comp}
        \STATE Update $A_p \leftarrow A_p \cup \{g^{*}\}$ and $U \leftarrow U \setminus \{g^{*}\}$.
     \ELSE
        \STATE Update $A_r \leftarrow A_r \cup \{c^{*}\}$ and $U \leftarrow U \setminus \{c^{*}\}$.
     \ENDIF
  \ENDWHILE
  \RETURN Allocation $\alloc=(A_1,\hdots,A_n)$
  \end{algorithmic}
\end{algorithm}

\begin{restatable}{theorem}{TheoremTwoAgents}
\label{theorem:two-agents}
Given any fair division instance $(N, M, \mcV)$ with two agents that have additive objective valuations, Algorithm \ref{algo:TwoAgents} computes an \eqx{} allocation in polynomial time.
\end{restatable}

\begin{proof}
We provide a proof of the theorem via induction on the number of allocated items. Initially, the algorithm starts with an empty allocation, which is trivially \eqx{}. Now, fix any iteration of the algorithm and let $\alloc$ be the partial allocation maintained at the beginning of the iteration. Partial allocation $\alloc$ is \eqx{} by the induction hypothesis. Also, for $p$ and $r$ as defined in the considered iteration, we have $v_p(A_p) \le v_r(A_r)$.

Note that if, in the iteration, item $g^*$ is assigned to agent $p$, then any chore $c$ assigned previously (i.e., in an earlier iteration) to the other agent $r$ must have greater absolute value. This follows from the algorithm's selection criterion (Line \ref{line:abs-comp}) and the fact that when chore $c$ is assigned to agent $r$, the two agents $p$ and $r$ were similarly the poorer and the richer one, respectively. That is, $|v_p(g^*)| \le |v_r(c)|$ for all chores $c \in A_r$. Similarly, if $c^*$ is assigned to agent $r$ in the considered iteration, any good $g$ assigned previously to the other agent $p$ must have greater absolute value, $|v_r(c^*)| \leq |v_p(g)|$ for all goods $g \in A_p$.\footnote{Recall that in the current context the two agents have objective valuations, i.e., the classification of items as goods or chores is consistent under the two valuations $v_1$ and $v_2$.} 

In the considered iteration, either good $g^*$ gets assigned to agent $p$ or chore $c^*$ is assigned to agent $r$. Next, we complete the inductive argument considering these cases. \\

\noindent
{\it Case {\rm I}a : Good $g^*$ is assigned to agent $p$ and $v_p(A_p \cup g^*) \le v_r(A_r)$.} By the induction hypothesis, for each chore $c \in A_p$, we have, $v_p(A_p \setminus \{ c \}) \ge v_r(A_r)$. Hence, $v_p(A_p \cup g^* \setminus \{ c \}) \ge v_r(A_r)$, i.e., there is no chore violation towards \eqx{}. Also, by the induction hypothesis, for each good $g \in A_r$, it holds that $ v_r(A_r \setminus \{ g\}) \leq v_p(A_p) \leq v_p(A_p \cup g^*)$. Hence, in this case, \eqx{} continues to hold even after the assignment of $g^*$.  \\

\noindent
{\it Case {\rm I}b: Good $g^*$ is assigned to agent $p$ and $v_p(A_p \cup g^*) > v_r(A_r)$.} Since the agents select the goods greedily, each good in $A_p$ has value at least $v_p(g^*)$. Hence, for each good $g \in A_p$, the following inequality holds: $v_p(A_p \cup g^* \setminus \{g\}) \le  v_p(A_p) \leq v_r(A_r)$. We obtain that there are no goods violation towards \eqx{}. Next, recall that for each chore $c \in A_r$ the inequality $|v_r(c)| \geq v_p(g^*)$ holds. This inequality leads to the desired bound for each chore $c \in A_r$:
\begin{align*}
v_r(A_r \setminus \{ c \}) & = v_r(A_r) + |v_r(c)| \geq v_p(A_p) + v_p(g^*) \\
& = v_p(A_p \cup g^*).
\end{align*} 
Hence, there are no chore violations as well and, in the current case, the updated allocation is \eqx{}. \\

\noindent
{\it Case {\rm II}a : Chore $c^*$ is assigned to agent $r$ and $v_p(A_p) \ge v_r(A_r \cup c^*)$.} Here, note that each chore $c$ previously selected by agent $r$ has absolute value at least $|v_r(c^*)|$. This bound gives us $v_p(A_p) \leq  v_r(A_r) \leq v_r(A_r \cup c^* \setminus \{c\})$, for all chores $c \in A_r$. Hence,  there are no chore violations in the updated allocation. Also, as mentioned previously, $|v_r(c^*)| \leq  v_p(g) $ for all goods $g \in A_p$. Therefore, we obtain that the updated allocation upholds the \eqx{} criterion with respect to goods' removal:
\begin{align*}
v_p(A_p \setminus \{ g \}) & = v_p(A_p) - v_p(g) \le v_r(A_r) + v_r(c^*) \\ 
& = v_r(A_r \cup c^*).
\end{align*}

\noindent
{\it Case {\rm II}b : Chore $c^*$ is assigned to agent $r$ and $v_p(A_p) < v_r(A_r \cup c^*)$.} The induction hypothesis gives us $v_r(A_r \setminus \{g\}) \leq v_p(A_p)$, for each good $g \in A_r$. Hence, there are no goods violation in the updated allocation: $ v_r(A_r \cup c^* \setminus \{ g \}) \leq v_r(A_r \setminus \{g\}) \leq v_p(A_p)$. The induction hypothesis also implies that $v_p(A_p \setminus \{ c\}) \ge v_r(A_r)$ for each chore $c \in A_p$. Hence, there are no chore violations in the resulting allocation: $v_p(A_p \setminus \{ c\}) \ge v_r(A_r) \geq v_r(A_r \cup c^*)$. 

Overall, we obtain that \eqx{} criterion is maintained as an invariant as the algorithm iterates. This completes the proof. 

\end{proof}

\subsection{Identically-Valued Chores}
\label{sec:IdenticalChores}
This section shows that \eqx{} allocations are guaranteed to exist when the agents have additive, objective valuations and they value the chores $c$ identically, i.e., $v_i(c) = v_j(c)$ for all agents $i, j \in N$.  

As noted, $C$ denotes the set of chores, and $G$ is the set of goods. For each chore $c \in C$, we will write $v_c < 0 $ to denote the common value of the chore among the agents.  

For intuition for the proof, consider an allocation $\alloc=(A_1, \ldots, A_n)$ that is not \eqx. There are two possibilities. There could be a chores violation --- there exist agents $i$ and $j$ and a chore $c \in A_i$ such that $v_i(A_i \setminus \{ c\}) < v_j(A_j)$. Here $c$ is called the violating chore. Or, there could be a goods violation --- there exist agents $i$ and $j$ and a good $g \in A_j$ such that $v_i(A_i) < v_j(A_j \setminus \{g\})$. Here $g$ a violating good. 

Our proof is based on the observation that transferring $c$ to $A_j$ in case of a chores violation and $g$ to $A_i$ under a goods violation leads to a lexicographic improvement in the tuple of values. The lexicographic order here additionally incorporates the sizes of the assigned bundles and the agents' indices for tie-breaking. 

Specifically, for any allocation $\mathcal{X} = (X_1, \ldots, X_n)$, we define permutation (ordering) $\sigma^{\mathcal{X}} \in \mathbb{S}_n$ over the $n$ agents such that 
\begin{itemize}
\item[(i)] Agents $i$ with lower values, $v_i(X_i)$, receive lower indices in $\sigma^{\mathcal{X}}$. 
\item[(ii)] Among agents with equal values, agents $i$ with lower bundle sizes, $|A_i|$, receive lower indices in $\sigma^{\mathcal{X}}$. 
\item[(iii)] Then, among agents with equal values and equal number of items, we order by index $i$.  
\end{itemize}

Plaut and Roughgarden~\cite{PlautR20} used a similar idea in their proof of the existence of \efx{} allocations for \emph{identical} valuations. They noted, via an example, that the lexicographic order over just the tuple of values does not suffice and a strengthening which incorporates bundle sizes is required. In particular, they introduced the $\precplus$ comparison operator, which we detail in Algorithm \ref{alg:Precplus}. Unlike their work, we use the $\precplus$ operator to obtain existential guarantees for \eqx{} allocations even when the valuations are nonidentical over the goods (but identical for chores).  

Note that if two agents have the same value and the same number of items, then Plaut and Roughgarden's results hold for any arbitrary but consistent tie-breaking between agents. In the definition of $\precplus$ (see Algorithm \ref{alg:Precplus}), we use the index of each agent $i \in [n]$ as a tie-breaker to compare two agents that are otherwise identical (have the same value for their own bundles and the same number of items). 

\begin{algorithm}[!htp]
\caption{$\precplus$ comparison operator}
\label{alg:Precplus}
\textbf{Input}: Allocations $\alloc$ and $\mathcal{B}$. \\
\textbf{Output}: \textbf{True} if $\alloc \precplus \mathcal{B}$, else return \textbf{False}. \\
\vspace*{-10pt}
\begin{algorithmic}[1] %[1] enables line numbers
\STATE Let $\sigma^\alloc$ and $\sigma^{\mathcal{B}}$ be the defined permutations of the agents associated with allocations $\alloc$ and $\mathcal{B}$, respectively. 

\FORALL{$\ell \in [n]$}
\STATE Set $i = \sigma^\mathcal{A}(\ell)$ and set $j = \sigma^{\mathcal{B}}(\ell)$. \COMMENT{$i$ and $j$ are the $\ell$th agents in the permutations.}
\IF {$v_i(A_i) \neq v_j(B_j)$} 
\RETURN $v_i(A_i) < v_j(B_j)$.
\ELSIF {$|A_i| \neq |B_j|$} 
\RETURN $|A_i| < |B_j|$.
\ELSIF{$i \neq j$}
\RETURN $i < j$.
\ENDIF
\ENDFOR
\end{algorithmic}
\end{algorithm}

\begin{theorem}[Theorem 4.1,~\cite{PlautR20}]
The comparison operator $\precplus$ induces a total order on the set of all allocations.
\label{thm:LeximinGivesTotalOrder}
\end{theorem}

We say an allocation $\alloc$ is a \lmplus{} allocation if $\mathcal{B} \precplus \alloc$ for all allocations $\mathcal{B}$. Theorem \ref{thm:LeximinGivesTotalOrder} ensures that a \lmplus{} allocation is guaranteed to exist. 

The following theorem then establishes the existential guarantee for \eqx{} allocations in the case of identically-valued chores. %The proof of Theorem \ref{thm:LeximinIsEQx} appears in Appendix~\ref{appendix:id-chore}.

\begin{restatable}{theorem}{TheoremLeximinEQx}
In a fair division instance if the agents have additive, objective valuations and they value the chores identically, then the \lmplus{} allocation is \eqx.
\label{thm:LeximinIsEQx}
\end{restatable}

\begin{proof}
Write $\alloc$ to denote a \lmplus{} allocation in the given instance; such an allocation necessarily exists (Theorem \ref{thm:LeximinGivesTotalOrder}).
Assume, towards a contradiction, that allocation $\alloc$ is not \eqx. Then, under $\alloc$, we either have a goods violation or a chores violation. 

We consider these two cases separately and show that, in both the cases, allocation $\alloc$ is not \lmplus. In particular, the allocation $\mathcal{B}$ obtained by transferring the violating good or the violating chore satisfies $\alloc \precplus \mathcal{B}$. \\

\noindent
{\it Goods Violation:} Here, let $i$ and $j$ be the agents with the property that  $v_i(A_i) < v_j(A_j \setminus \{ g \})$ for a violating good $g \in A_j$. Now, consider the allocation $\mathcal{B}=(B_1, \ldots, B_n)$ obtained by transferring $g$ to agent $i$. That is, let $B_i = A_i \cup \{g\}$ and $B_j = A_j \setminus \{g\}$; the bundles of all the other agents remain unchanged. Note that $v_i(B_i) \ge v_i(A_i)$ and $|B_i| > |A_i|$, and since $g$ was a violating good, $v_j(B_j) > v_i(A_i)$.

Let $\ell$ be the index of agent $i$ in the permutation $\sigma^\alloc$, i.e., $\sigma^\alloc(\ell) = i$. Note first that the first $(\ell-1)$ agents remain unchanged in $\sigma^\alloc$ and $\sigma^{\mathcal{B}}$ -- these are agents with values, bundle sizes, and indices---in that order---lower than that of $i$. Now, if the $\ell$th agent in permutation $\sigma^{\mathcal{B}}$ is agent $i$ itself, then $\alloc \precplus \mathcal{B}$. This follows from the fact that $v_i(B_i) \geq v_i(A_i)$ and $|B_i| > |A_i|$. In case the $\ell$th agent in $\sigma^{\mathcal{B}}$ is $j$, then again $\alloc \precplus \mathcal{B}$; recall that $v_j(B_j) > v_i(A_i)$. Finally, if the $\ell$th agent, $\sigma^{\mathcal{B}}(\ell)$, is neither $i$ nor $j$, then this agent was placed after agent $i$ in $\sigma^\alloc$. Since the bundle assigned to the agent $\sigma^{\mathcal{B}}(\ell)$ is unchanged, it has either a higher value, more items, or a larger index than $i$. In this case as well, $\alloc \precplus \mathcal{B}$. \\

\noindent
{\it Chores Violation:} Here, let $i$ and $j$ be the agents with the property that $v_i(A_i \setminus \{ c \}) < v_j(A_j)$ for some chore $c \in A_i$. Now, consider the allocation $\mathcal{B}=(B_1, \ldots, B_n)$ obtained by transferring chore $c$ to agent $j$. That is, $B_i = A_i \setminus \{c\}$ and $B_j = A_j \cup \{c\}$. Recall that the common value of the chore $v_c < 0$. Hence, the following strict inequality holds $v_i(B_i) > v_i(A_i)$. Also, since $c$ is a violating chore, we have $v_j(A_j) > v_i(A_i) - v_c$. The last inequality reduces to $v_j(B_j) = v_j(A_j) + v_c > v_i(A_i)$.\footnote{This is where we require that the chores be identical; the value of $c$ does not change when it is transferred from agent $i$ to agent $j$.} 

Again, let $\ell$ be the index of agent $i$ in the permutation $\sigma^\alloc$, i.e., $\sigma^\alloc(\ell) = i$. As before, the first $(\ell-1)$ agents remain unchanged in $\sigma^\alloc$ and $\sigma^{\mathcal{B}}$. If the $\ell$th agent in $\sigma^{\mathcal{B}}$ is agent $i$ itself, then, using the above-mentioned strict inequality, $v_i(B_i) > v_i(A_i)$, we obtain $\alloc \precplus \mathcal{B}$. In case the $\ell$th agent in $\sigma^{\mathcal{B}}$ is $j$, then again $\alloc \precplus \mathcal{B}$, since $v_j(B_j) > v_i(A_i)$. Finally, if the $\ell$th agent is neither $i$ nor $j$, then this agent was placed after agent $i$ in $\sigma^\alloc$. Since the bundle assigned to this agent is unchanged, it has either a higher value, more items, or a larger index than $i$. Overall, we obtain the desired lexicographic dominance, $\alloc \precplus \mathcal{B}$. 
\end{proof}

\subsection{Additive Goods and a Single Chore}
\label{sec:singlechore}

We now show that in instances with additive objective valuations and a single chore $c$, \eqx{} allocations exist. Note first that in the presence of even a single non-identical chore, the \lmplus{} allocation may no longer be \eqx{} (Appendix \ref{appendix:Leximin-Limitation}). Instead, our proof of existence is based on a careful analysis of a version of local search, and keeping track of the movement of the single chore.  

The local search algorithm (Algorithm \ref{alg:SingleChore}) proceeds as follows. Write $\alloc$ to denote the current allocation. Further, let $p$ be the agent that appears first in the permutation $\sigma^\alloc$ (as defined in the previous subsection) and $r$ be the agent that appears last. Note that $p$ and $r$ are a poorest and a richest agent, respectively. If $\alloc$ is not \eqx{}, then there must be either a goods violation or a chores violation. 

\begin{algorithm}[!htp]
\caption{Algorithm for single-chore setting}
  \label{alg:SingleChore}
  \textbf{Input:} Fair division instance $(N,M,\mathcal{V})$ with additive, objective valuations and a single chore $c$\\
  \textbf{Output:} \eqx{} allocation $\alloc$
  \begin{algorithmic}[1]  
  \STATE Initialize $A_1 = M$ and $A_i = \emptyset$ for all agents $i \neq 1$.  
  \WHILE{$\alloc=(A_1, \ldots, A_n)$ is \NOT \eqx{}}
	\STATE $p = \sigma^\alloc(1)$.
	\COMMENT{$p$ is the first agent according to $\sigma^\alloc$.}
	\WHILE {there exists agent $i \in N$ and good $g \in A_i$ such that $v_i(A_i \setminus \{ g \}) > v_p(A_p)$} \label{line:SCGVBegin}
		\STATE {Update $A_p \gets A_p \cup \{g\}$ and $A_i \gets A_i \setminus \{g\}$}.
		\STATE Update $p = \sigma^\alloc(1)$.
	\ENDWHILE \label{line:SCGVEnd}
	\COMMENT{After resolving all goods violations, check for chores violation.}
	\STATE $r \gets \sigma^\alloc(n)$.
	\COMMENT{$r$ is the last agent according to $\sigma^\alloc$.}
	\IF {$v^+(\alloc) < v_r(A_r)$} \label{line:CheckChoresViolation}
		\STATE {Update $A_{\kappa(\alloc)} \gets A_{\kappa(\alloc)} \setminus \{c\}$ and $A_r \gets A_r \cup \{c\}$}.
	\ENDIF
  \ENDWHILE
   \RETURN $\alloc$
  \end{algorithmic}
\end{algorithm}

If there is a goods violation in $\alloc$, then there exists an agent $i$ and good $g \in A_i$ with the property that $v_i(A_i \setminus \{g\}) > v_p(A_p)$. In our algorithm, we resolve the goods violation by transferring good $g$ to agent $p$. The algorithm resolves all goods violations, before addressing the chores violation. Now, if there is a chores violation, then for some some agent $i$ and the chore $c \in A_i$ it holds that $v_i(A_i \setminus \{ c \}) < v_r(A_r)$. The algorithm resolves the chores violation by transferring chore $c$ to agent richest agent $r$. The algorithm repeats these steps---resolving all goods violations and then the chores violation---until there are no more violations.  

By design, the algorithm terminates only when the maintained allocation is \eqx{}. To prove that the algorithm indeed terminates, we need the following notation. For an allocation $\alloc$, we denote by $\kappa(\alloc)$ as the agent holding the unique chore $c$, i.e., $c \in A_{\kappa(\alloc)}$. We call the value $v^+(\alloc) := v_{\kappa(\alloc)}(A_{\kappa(\alloc)} \setminus c)$ as the \emph{cutoff value} of the allocation $\alloc$ and of the agent ${\kappa(\alloc)}$. 

In case of a chores violation we have $v_{\kappa(\alloc)}(A_{\kappa(\alloc)} \setminus c) < v_r(A_r)$. Write $\mathcal{B}$ to denote the allocation after transferring chore $c$ to agent $r$. Note that such a chore transfer increases the cutoff value: under allocation $\alloc$, the cutoff value is $v^+(\alloc) = v_{\kappa(\alloc)}(A_{\kappa(\alloc)} \setminus c)$ and, under the updated allocation $\mathcal{B}$, it is $v^+(\mathcal{B}) = v_r(A_r)$. In particular, $v^+(\alloc) < v^+(\mathcal{B})$. In our analysis, we will keep track of two progress measures separately: the lexicographic value of the allocation (according to the order $\precplus$ defined in Section \ref{sec:IdenticalChores}) and the cutoff value of the allocation. We will show, in particular, that the cutoff value of the allocation is nondecreasing between successive chores violations, and strictly increases whenever a chores violation is resolved. Whenever a goods violation is resolved, we obtain a lexicographic improvement in the allocation; a chores violation may however result in a lexicographic decrease. Since both the cutoff value and the lexicographic value can only increase a finite number of times, the local search algorithm must terminate in finite time. Formally,  

\begin{restatable}{theorem}{TheoremSingleChore}
Given any fair division instance with additive objective valuations and a single chore, Algorithm~\ref{alg:SingleChore} terminates in finite time and returns an \eqx{} allocation.
\label{thm:SingleChore}
\end{restatable}
\begin{proof}
The outer while-loop terminates only when the current allocation is \eqx{}. Hence, if the algorithm terminates, it must return an \eqx{} allocation. It remains to show that that the algorithm terminates in finite time.

Firstly, consider the inner while-loop, which resolves goods violations. It follows from the proof of Theorem~\ref{thm:LeximinIsEQx} that resolving a goods violation leads to a lexicographic improvement. That is, if $\alloc$ is the allocation before the inner while-loop executes and $\mathcal{B}$ is the allocation after the inner while-loop terminates, then $\alloc \precplus \mathcal{B}$. In particular, $\min_i v_i(A_i) \le \min_i v_i(B_i)$.

Now consider two successive executions of Line~\ref{line:CheckChoresViolation}. Let $\alloc'$ be the allocation just before the first execution of Line~\ref{line:CheckChoresViolation}, $\alloc$ be the allocation after the chores violation is resolved, and $\mathcal{B}$ be the allocation just before the second execution of Line~\ref{line:CheckChoresViolation} (i.e., after the inner while loop terminates). We claim that $v^+(\alloc') < v^+(\mathcal{B})$.

To prove the claim, we first show that $v^+(\alloc')  < v^+(\alloc)$. Towards this, note that if there is a chores violation in $\alloc'$, then the chore is moved from agent ${\kappa(\alloc')}$ to agent ${\kappa(\alloc)}$. Further, $v_{\kappa(\alloc')}(A'_{\kappa(\alloc')}\setminus \{  c \}) < v_{\kappa(\alloc)}(A'_{\kappa(\alloc)})$, since $\kappa(\alloc)$ is agent $r$ in $\alloc'$. Then, $v^+(\alloc')  < v^+(\alloc)$ by definition of the cutoff value.

We now show that $v^+(\alloc) \le v^+(\mathcal{B})$. Write $k := \kappa(\alloc)$, and note that agent $k$ holds the chore in $\alloc$, $\mathcal{B}$, and all the  interim allocations during the execution of the inner while loop. We consider two cases. 

In the simpler first case, agent $k$ is a poorest agent in $\alloc$, after receiving the chore. Then, as mentioned earlier, $\alloc \precplus \mathcal{B}$ and, in particular, the minimum value is nondecreasing. Hence, $v_k(A_k) \le v_k(B_k)$. Since the cutoff value is obtained by removing the chore from $A_k$ and $B_k$, it follows that $v^+(\alloc) \le v^+(\mathcal{B})$.

In the second case, agent $k$ is not a poorest agent in $\alloc$. Since the previous allocation $\alloc'$ had no goods violations, removal of any good $g \in A_k'$ made $k$ the poorest agent in $\alloc'$. This property continues to hold after the chore is transferred to agent $k$. That is, $v_k(A_k \setminus \{ g \}) \le v_i(A_i)$ for each good $g \in A_k$ and all agents $i \neq k$. 

Let $T$ denote the number of iterations of the inner while-loop when it updates allocation $\alloc$ to allocation $\mathcal{B}$. Write $\alloc^i$ to denote the allocation after the $i$th execution of the inner while loop. Then $\alloc = \alloc^0$, and $\mathcal{B} = \alloc^T$.

Let $t$ be the first iteration of the inner while loop where a good is removed from agent $k$'s bundle (this is a necessary condition for agent $k$'s value to decrease). However as we have noted, in allocation $\alloc$, removing any good from agent $k$'s bundle made it the poorest agent. If agent $k$'s bundle is unchanged prior to this execution, i.e., $A_k^{t-1} = A_k$, then since the minimum value is nondecreasing, removing any good from agent $k$'s bundle would still make it the poorest agent. Hence, we get that $A_k^{t-1} \supsetneq A_k$, and there must exist some $t' < t$ such that a good was added to agent $k$'s bundle in execution $t'$. 

Now let $t'$ be the first execution of the inner while loop where a good is added to agent $k$'s bundle. For a good to be added to agent $k$'s bundle, agent $k$ must be the poorest agent prior to the addition, i.e., $v_k(A_k^{t'-1}) = \min_i v_i(A_i^{t'-1})$. Note that this is the first modification of agent $k$'s bundle by the inner while loop, hence $A_k^{t'-1} = A_k$. But we know that the minimum value is nondecreasing in the inner while loop. Hence, we obtain that
\begin{align*}
v_k(B_k) \ge \min_i v_i(B_i) \ge \min_i v_i(A_i^{t'-1}) \\
	\qquad = v_k(A_k^{t'-1}) = v_k(A_k).
\end{align*}

This completes the proof of the claim $v^+(\alloc) \le v^+(\mathcal{B})$ since agent $k$ holds the chore in allocations $\alloc$ and $\mathcal{B}$. 

Therefore, we obtain that between two successive executions of Line~\ref{line:CheckChoresViolation}, the cutoff value of the allocation must strictly increase. This can clearly happen a finite number of times. After the last execution of Line~\ref{line:CheckChoresViolation}, after resolving goods violations in the inner while-loop (which must terminate, since lexicographic improvements can only occur a finite number of times), there cannot be any chores violations, and the outer while-loop must terminate and return an \eqx{} allocation.
\end{proof}

\section{Extensions}  
\label{sec:Extensions}

In this section, we show that all our positive results hold even if we replace goods with chores and vice versa. We skip formal proofs for these results. Instead, we outline the main differences encountered while establishing the converse results.

\subsection{Monotone valuations}

The results from Section~\ref{sec:Monotone} hold if the agents have \emph{monotone nonincreasing} valuations, instead of monotone nondecreasing as earlier. That is, our results hold if all items are chores, instead of goods. In this case, Algorithm~\ref{alg:AddandFix2} requires that instead of the poorest and second-poorest agents $p$ and $p'$, we consider the richest and second-richest agents $r$ and $r'$. The richest agent $r$ then selects the unallocated chore with the smallest marginal value, i.e., the item which decreases its value by the largest amount. The Add phase adds chores in this manner to agent $r$ as long as it remains a richest agent. The Fix phase removes chores from $A_r$ that, upon removal, do not make it a richest agent. With these changes, the algorithm computes an \eqx{} allocation in pseudo-polynomial time.

For polynomial time computation, instead of weakly well-layered valuations, we require that the agents' valuations be \emph{negatively weakly well-layered}:

\begin{definition}
\label{definition:negatively-well-layered}
    A valuation function $v:2^M \rightarrow \mathbb{Z}_{\le 0}$ is said to be \emph{negatively weakly well-layered} if for any set $M' \subseteq M$ the sets $S_0$, $S_1$, $\ldots$ obtained by the greedy algorithm (that is, $S_0 = \emptyset$ and $S_i = S_{i-1} \cup \{x_i\}$, where $x_i \in \argmin_{x \in M' \setminus S_{i-1}} v(S_{i-1} \cup x)$, for $i \le |M'|$) are optimal, in the sense that $v(S_i) = \min_{S \subseteq M': |S| = i} v(S)$ for all $i \le |M'|$.
\end{definition}

With these valuations, one can prove that Algorithm~\ref{alg:AddandFix2} indeed runs in polynomial time.

Lastly, for monotone nonincreasing valuations, given parameter $\varepsilon \in [0,1]$, an allocation $\alloc$ is said to be an $(1 + \varepsilon)$-\eqx{} allocation if for every pair of agents $i, j \in N$ and for each chore $c \in A_i$ we have $\ v_i(A_i \setminus \{ c \}) \geq (1+ \varepsilon) v_j(A_j)$. Hence, in an $(1+ \varepsilon)$-\eqx{} allocation, removing any chore from any agent $i$'s bundle improves $i$'s value to at least $(1+ \varepsilon)$ times the maximum (recall that all values are nonpositive).

We now modify Algorithm~\ref{alg:AddandFix2} as follows:

\noindent \ref{line:AddCondition}: \ \textbf{while} $v_r(A_r) \geq ( 1+ \varepsilon) v_{r'}(A_{r'})$ \textbf{and} $U \neq \emptyset$ \textbf{do}...

\noindent \ref{line:FixCondition}: \ \textbf{while} there exists $\widehat{c} \in A_r$ such that $v_r(A_r \setminus \{ \widehat{c} \}) < ( 1+ \varepsilon) v_{r'}(A_{r'})$ \textbf{do}...

One can prove that this modified algorithm returns a $ ( 1+ \varepsilon)$-\eqx{} allocation in ${O\left(\frac{m^2n}{\varepsilon} \log | V_{\min} |\right)}$ time, where $V_{\min} := \min_{i \in N} v_i(M)$.

\subsection{Additive nonmonotone valuations}

\paragraph*{Identically-valued goods.} We show that the positive results from Section~\ref{sec:IdenticalChores} can be obtained if we have identically-valued goods instead. In this case, agents have additive, objective valuations, and they value the goods $g$ identically, i.e., $v_i(g) = v_j(g)$ for all agents $i,j \in N$. 

In the earlier case when agents valued \emph{chores} identically, for an allocation $\mathcal{X} = (X_1, \ldots, X_n)$ we ordered agents by increasing value $v_i(X_i)$. In this case, when agents value \emph{goods} identically, we order agents by \emph{increasing negation of their values} $-v_i(X_i)$ (hence, the agent with largest value appears first in this order). Thus the permutation $\sigma^\mathcal{X}$ is defined as:
\begin{itemize}
\item[(i)] Agents with lower negated values, $-v_i(X_i)$, receive lower indices.
\item[(ii)] Among agents with equal values, agents $i$ with lower bundle sizes $|X_i|$ receive lower indices.
\item[(iii)] Agents with equal values and number of items are ordered by the index $i$. 
\end{itemize}

Then for allocations $\mathcal{A}$, $\mathcal{B}$, we say that $\mathcal{A} \precplus \mathcal{B}$ if, for the first index $\ell$ where they differ, if $i = \sigma^\mathcal{A}(\ell)$ and $j =  \sigma^\mathcal{B}(\ell)$, 
\begin{itemize}
\item[(i)] either $-v_i(A_i) < -v_j(B_j)$, 
\item[(ii)] or $-v_i(A_i) = -v_j(B_j)$ and $|A_i| < |B_j|$,
\item[(iii)] or $-v_i(A_i) = -v_j(B_j)$, $|A_i| = |B_j|$, and $i < j$.
\end{itemize}

Thus, an allocation $\alloc$ is a \lmplus{} allocation if $\mathcal{B} \precplus \alloc$ for all allocations $\mathcal{B}$. By Theorem \ref{thm:LeximinGivesTotalOrder}, we know that \lmplus{} allocation is guaranteed to exist. 

As in the earlier case, we can complete the proof for the current setting by using the observation that if an allocation  $\alloc$ is not EQx then resolving the violation gives a lexicographic improvement.

 \paragraph*{Additive chores and a single good.} Algorithm~\ref{alg:SingleChore} is modified in the natural manner, replacing goods with chores and vice versa.

\begin{algorithm}[!htp]
\caption{Algorithm for single good setting}
  \label{alg:SingleGood}
  \textbf{Input:} Fair division instance $(N,M,\mathcal{V})$ with additive, objective valuations and a single good $g$\\
  \textbf{Output:} \eqx{} allocation $\alloc$
  \begin{algorithmic}[1] %[1] enables line numbers
  \STATE Initialize $A_1 = M$ and $A_i = \emptyset$ for all agents $i \neq 1$.  
  \WHILE{$\alloc=(A_1, \ldots, A_n)$ is \NOT \eqx{}}
	\STATE $r = \sigma^\alloc(1)$.
	\COMMENT{$r$ is the first agent according to $\sigma^\alloc$.}
	\WHILE {there exists agent $i \in N$ and chore $c \in A_i$ such that $v_i(A_i \setminus \{ c \}) < v_r(A_r)$} \label{line:SCGVBegin}
		\STATE {Update $A_r \gets A_r \cup \{c\}$ and $A_i \gets A_i \setminus \{c\}$}.
		\STATE Update $r = \sigma^\alloc(1)$.
	\ENDWHILE \label{line:SCGVEnd}
	\COMMENT{After resolving all chores violations, check for goods violation.}
	\STATE $p \gets \sigma^\alloc(n)$.
	\COMMENT{$p$ is the last agent according to $\sigma^\alloc$.}
	\IF {$v^+(\alloc) > v_p(A_p)$} \label{line:CheckGoodsViolation}
		\STATE {Update $A_{\kappa(\alloc)} \gets A_{\kappa(\alloc)} \setminus \{g\}$ and $A_p \gets A_p \cup \{g\}$}.
	\ENDIF
  \ENDWHILE
   \RETURN $\alloc$
  \end{algorithmic}
\end{algorithm}

As before, in an allocation $\alloc$, $\kappa(\alloc)$ is the agent holding the unique \emph{good} $g$, and the value $v^+(\alloc) := v_{\kappa(\alloc)}(A_{\kappa(\alloc)} \setminus g)$ is the \emph{cutoff value}. In case of a goods violation we have $v_{\kappa(\alloc)}(A_{\kappa(\alloc)} \setminus g) > v_p(A_p)$. Let $\mathcal{B}$ denote the allocation after transferring good $g$ to agent $p$. Note that such a good transfer decreases the cutoff value: under allocation $\alloc$, the cutoff value is $v^+(\alloc) = v_{\kappa(\alloc)}(A_{\kappa(\alloc)} \setminus g)$ and, under the updated allocation $\mathcal{B}$, it is $v^+(\mathcal{B}) = v_p(A_p)$. In particular, $v^+(\alloc) > v^+(\mathcal{B})$. As before, in our analysis, we will keep track of two progress measures separately: the lexicographic value of the allocation (according to the order $\precplus$ defined for identically-valued goods and the cutoff value of the allocation. We can show that the cutoff value of the allocation is nonincreasing between successive goods violations, and strictly decreases whenever a goods violation is resolved. Whenever a chores violation is resolved, we obtain a lexicographic improvement in the allocation; a goods violation may however result in a lexicographic decrease. Since the cutoff value can decrease and the lexicographic value can increase only a finite number of times, the local search algorithm must terminate in finite time. 

\section{Conclusion and Future Work}
Our work resolves fundamental questions regarding the existence of \eqx{} allocations. We present sweeping positive results when all the items are goods \emph{or} all items are chores; this includes universal existence of \eqx{} allocations under general, monotone valuations and an accompanying pseudo-polynomial time algorithm. For monotone valuations, we also provide a fully polynomial-time approximation scheme (FPTAS) for finding approximately \eqx{} allocations. In addition, we show that under weakly well-layered valuations \eqx{} allocations can be computed efficiently.  
 
For mixed items (goods and chores), we show that \eqx{} allocations may not exist, and show both hardness results and polynomial time algorithms for determining the existence of an \eqx{} allocation. For additively-valued goods and chores, our results present a mixed picture: existence and efficient computation for two agents, and existence either when each chore or each good has the same value among the agents, or if there is a single chore or a single good.  

A number of significant open questions remain. First, the existence of \eqx{} allocations under objective, additive valuations remains unresolved. Second, efficient algorithms for computing \eqx{} allocations in this case (or even for goods and a single chore, or identical chores) appear challenging. Lastly, given the practical significance, truthful mechanisms for obtaining \eqx{} allocations may prove useful.

\paragraph*{Acknowledgments.} We thank Rohit Vaish for useful discussions, and suggestions regarding this paper.

\bibliographystyle{alpha}
\bibliography{References}

\begin{thebibliography}{CKKK12}

\bibitem[AD15]{AumannD15}
Yonatan Aumann and Yair Dombb.
\newblock The efficiency of fair division with connected pieces.
\newblock {\em ACM Transactions on Economics and Computation (TEAC)},
  3(4):1--16, 2015.

\bibitem[AR20]{AzizR20}
Haris Aziz and Simon Rey.
\newblock Almost group envy-free allocation of indivisible goods and chores.
\newblock In Christian Bessiere, editor, {\em Proceedings of the Twenty-Ninth
  International Joint Conference on Artificial Intelligence, {IJCAI} 2020},
  pages 39--45. ijcai.org, 2020.

\bibitem[BMSV23]{BhaskarMSV23}
Umang Bhaskar, Neeldhara Misra, Aditi Sethia, and Rohit Vaish.
\newblock The price of equity with binary valuations and few agent types.
\newblock {\em CoRR}, abs/2307.06726, 2023.

\bibitem[BNV23]{BarmanNV23}
Siddharth Barman, Vishnu~V. Narayan, and Paritosh Verma.
\newblock Fair chore division under binary supermodular costs.
\newblock In Noa Agmon, Bo~An, Alessandro Ricci, and William Yeoh, editors,
  {\em Proceedings of the 2023 International Conference on Autonomous Agents
  and Multiagent Systems, {AAMAS} 2023, London, United Kingdom, 29 May 2023 - 2
  June 2023}, pages 2863--2865. {ACM}, 2023.

\bibitem[CDP13]{CechlarovaDP13}
Katar{\'i}na Cechl{\'a}rov{\'a}, Jozef Dobo{\v{s}}, and Eva Pill{\'a}rov{\'a}.
\newblock On the existence of equitable cake divisions.
\newblock {\em Information Sciences}, 228:239--245, 2013.

\bibitem[Ch{\`e}17]{Cheze17}
Guillaume Ch{\`e}ze.
\newblock Existence of a simple and equitable fair division: A short proof.
\newblock {\em Mathematical Social Sciences}, 87:92--93, 2017.

\bibitem[CKKK12]{CKK+12efficiency}
Ioannis Caragiannis, Christos Kaklamanis, Panagiotis Kanellopoulos, and Maria
  Kyropoulou.
\newblock {The Efficiency of Fair Division}.
\newblock {\em Theory of Computing Systems}, 50(4):589--610, 2012.

\bibitem[CL20]{ChenL20}
Xingyu Chen and Zijie Liu.
\newblock The fairness of leximin in allocation of indivisible chores.
\newblock {\em CoRR}, abs/2005.04864, 2020.

\bibitem[CP12]{CechlarovaP12}
Katar{\'\i}na Cechl{\'a}rov{\'a} and Eva Pill{\'a}rov{\'a}.
\newblock On the computability of equitable divisions.
\newblock {\em Discrete Optimization}, 9(4):249--257, 2012.

\bibitem[DS61]{DubinsS61}
Lester~E Dubins and Edwin~H Spanier.
\newblock How to cut a cake fairly.
\newblock {\em The American Mathematical Monthly}, 68(1P1):1--17, 1961.

\bibitem[Fol66]{foley1966resource}
Duncan~Karl Foley.
\newblock {\em Resource allocation and the public sector}.
\newblock Yale University, 1966.

\bibitem[FSVX19]{FSV+19equitable}
Rupert Freeman, Sujoy Sikdar, Rohit Vaish, and Lirong Xia.
\newblock {Equitable Allocations of Indivisible Goods}.
\newblock In {\em Proceedings of the 28th International Joint Conference on
  Artificial Intelligence}, pages 280--286, 2019.

\bibitem[FSVX20]{FSV+20equitable}
Rupert Freeman, Sujoy Sikdar, Rohit Vaish, and Lirong Xia.
\newblock {Equitable Allocations of Indivisible Chores}.
\newblock In {\em Proceedings of the 19th International Conference on
  Autonomous Agents and MultiAgent Systems}, pages 384--392, 2020.

\bibitem[GHH23]{GoldbergHH23}
Paul~W. Goldberg, Kasper H{\o}gh, and Alexandros Hollender.
\newblock The frontier of intractability for {EFX} with two agents.
\newblock {\em CoRR}, abs/2301.10354, 2023.
\newblock {T}o appear in {SAGT '23}.

\bibitem[GMPZ17]{GalMPZ17}
Ya'akov~(Kobi) Gal, Moshe Mash, Ariel~D. Procaccia, and Yair Zick.
\newblock Which is the fairest (rent division) of them all?
\newblock {\em J. {ACM}}, 64(6):39:1--39:22, 2017.

\bibitem[GMT14a]{GMT14near}
Laurent Gourv{\`e}s, J{\'e}r{\^o}me Monnot, and Lydia Tlilane.
\newblock {Near Fairness in Matroids}.
\newblock In {\em Proceedings of the 21st European Conference on Artificial
  Intelligence}, pages 393--398, 2014.

\bibitem[GMT14b]{GourvesMT14}
Laurent Gourv{\`{e}}s, J{\'{e}}r{\^{o}}me Monnot, and Lydia Tlilane.
\newblock Near fairness in matroids.
\newblock In Torsten Schaub, Gerhard Friedrich, and Barry O'Sullivan, editors,
  {\em {ECAI} 2014 - 21st European Conference on Artificial Intelligence, 18-22
  August 2014, Prague, Czech Republic - Including Prestigious Applications of
  Intelligent Systems {(PAIS} 2014)}, volume 263 of {\em Frontiers in
  Artificial Intelligence and Applications}, pages 393--398. {IOS} Press, 2014.

\bibitem[GP15]{GP15spliddit}
Jonathan Goldman and Ariel~D Procaccia.
\newblock {Spliddit: Unleashing Fair Division Algorithms}.
\newblock {\em ACM SIGecom Exchanges}, 13(2):41--46, 2015.

\bibitem[HP09]{HerreinerP09}
Dorothea~K Herreiner and Clemens~D Puppe.
\newblock Envy freeness in experimental fair division problems.
\newblock {\em Theory and decision}, 67:65--100, 2009.

\bibitem[HP10]{HerreinerP10}
Dorothea~K Herreiner and Clemens Puppe.
\newblock Inequality aversion and efficiency with ordinal and cardinal social
  preferences—an experimental study.
\newblock {\em Journal of Economic Behavior \& Organization}, 76(2):238--253,
  2010.

\bibitem[Kag21]{Divorce}
Julia Kagan.
\newblock {Equitable Distribution: Definition, State Laws, Exempt Property}.
\newblock \url{https://www.investopedia.com/terms/e/equitable-division.asp},
  2021.
\newblock Accessed: 2023-08-12.

\bibitem[Mou04]{moulin2004fair}
Herv{\'e} Moulin.
\newblock {\em Fair division and collective welfare}.
\newblock MIT press, 2004.

\bibitem[PR20a]{PR20almost}
Benjamin Plaut and Tim Roughgarden.
\newblock {Almost Envy-Freeness with General Valuations}.
\newblock {\em SIAM Journal on Discrete Mathematics}, 34(2):1039--1068, 2020.

\bibitem[PR20b]{PlautR20}
Benjamin Plaut and Tim Roughgarden.
\newblock Almost envy-freeness with general valuations.
\newblock {\em {SIAM} J. Discret. Math.}, 34(2):1039--1068, 2020.

\bibitem[PW17]{ProcacciaW17}
Ariel~D Procaccia and Junxing Wang.
\newblock A lower bound for equitable cake cutting.
\newblock In {\em Proceedings of the 2017 ACM Conference on Economics and
  Computation}, pages 479--495, 2017.

\bibitem[SCD23]{SunCD23}
Ankang Sun, Bo~Chen, and Xuan~Vinh Doan.
\newblock Equitability and welfare maximization for allocating indivisible
  items.
\newblock {\em Autonomous Agents and Multi-Agent Systems}, 37(1):8, 2023.

\bibitem[Var74]{V74equity}
Hal~R Varian.
\newblock {Equity, Envy, and Efficiency}.
\newblock {\em Journal of Economic Theory}, 9(1):63--91, 1974.

\end{thebibliography}

\appendix
\section{Missing Proofs from Section \ref{section:monotone-apx}}
\label{appendix:MonApx}

In this section, we first state and prove Claim \ref{claim:AddandFixApproximation}. The claim is then used to establish Theorem \ref{thm:MonotoneApproximation}. 

%\ClaimAddFixApx*
\begin{restatable}{claim}{ClaimAddFixApx}
After every outer iteration, either (i) $p$ is no longer the poorest agent, and the value of agent $p$ has increased multiplicatively by at least a factor $\left(\frac{1}{1-\varepsilon}\right)$, or (ii) all remaining unassigned goods are assigned to agent $p$ and the algorithm terminates.
\label{claim:AddandFixApproximation}
\end{restatable}
\begin{proof}
At the start of the outer iteration, $v_p(A_p) \le v_{p'}(A_{p'})$. The Add phase terminates only when either all remaining goods are assigned to agent $p$, or $v_p(A_p) > \left(\frac{1}{1-\varepsilon} \right) v_{p'}(A_{p'})$. Hence, at the end of the Add phase, either $v_p(A_p) > \left(\frac{1}{1-\varepsilon} \right) v_{p'}(A_{p'})$, or all remaining goods are assigned to agent $p$ and $v_p(A_p) \le \left(\frac{1}{1-\varepsilon} \right) v_{p'}(A_{p'})$. In the latter case, the Fix phase is not executed, and the algorithm terminates as claimed. In the former case, the Fix phase may be executed. Note, however, that a good $\widehat{g}$ is removed from agent $p$ only if $v_p \left(A_p \setminus \{ \widehat{g} \} \right) > \left(\frac{1}{1-\varepsilon} \right)  v_{p'}(A_{p'})$. Hence, after termination of the Fix phase, we must have that $v_p(A_p) > \left(\frac{1}{1-\varepsilon} \right) v_{p'}(A_{p'})$. At this point, the value of agent $p$ has strictly increased to above that of agent $p'$, and by a factor of $\left(\frac{1}{1-\varepsilon} \right)$, as claimed.
\end{proof}

\TheoremMonApx*
\begin{proof}
For the time complexity, note that, from Claim~\ref{claim:AddandFixApproximation}, the value of an agent $i$ can increase at most $\log_{\left( \frac{1}{1-\varepsilon}\right)} v_i(M) \le \frac{\log V_{\max}}{\varepsilon}$ times. Hence, the number of outer iterations is at most $O \left((n \log V_{\max})/ \varepsilon \right)$. As noted in the proof of Theorem~\ref{thm:Monotone}, the runtime of each outer iteration is $O(m^2)$. The algorithm thus terminates in $O \left(\frac{m^2n}{\varepsilon}\log V_{\max} \right)$ time. 

To prove that the allocation computed by the algorithm is indeed $(1- \varepsilon)$-\eqx{}, we proceed as before by induction on the number of outer iterations. Initially, the allocation is empty, which is trivially \eqx{} (and, hence, $(1-\varepsilon)$-\eqx{}). Since the bundles assigned to agents, other than $p$, remain unchanged in an outer iteration, any $(1-\varepsilon)$-\eqx{} violation must involve agent $p$. Let $\mathcal{B}$ be the allocation obtained after an outer iteration. Then $v_i(B_i) = v_i(A_i)$ for $i \neq p$. Also, Claim~\ref{claim:AddandFixApproximation} ensures that $v_p(B_p) \ge v_p(A_p)$.

To show that allocation $\mathcal{B}=(B_1, \ldots, B_n)$ is $(1-\varepsilon)$-\eqx{}, we need to show that for any agent $i \in N$, the removal of any good $g \in B_i$ 
reduces the value of the bundle to at most $\frac{1}{1-\varepsilon}$ times the poorest agent. This condition holds---via the induction hypothesis---for all agents $i \neq p$; recall that $B_i = A_i$, for all $i \neq p$, and $v_p(B_p) \ge v_p(A_p)$. 

For agent $p$, note that after the completion of the Fix phase, the removal of any good from $p$'s bundle reduces its value to at most $\left(\frac{1}{1-\varepsilon} \right) v_{p'}(A_{p'}) = \left(\frac{1}{1-\varepsilon} \right) v_{p'}(B_{p'})$. Since $p$ was the poorest and $p'$ was the second poorest agent in allocation $\alloc$, this implies that the removal of any good from $B_p$ brings down $p$'s value to below $\frac{1}{1- \varepsilon}$ times the minimum: $v_p(B_p \setminus \{ g\}) \leq \left(\frac{1}{1-\varepsilon} \right) v_j(B_j)$ for all agents $j$ and each good $g \in B_p$. 

The theorem stands proved. 
\end{proof}

\section{Limitation of Leximin++ under Nonmonotone Valuations}
\label{appendix:Leximin-Limitation}

This section shows that, even in the presence of a single, non-identically valued chore, leximin allocations are not guaranteed to be \eqx{}. Recall that, by contrast and for identical chores, a \lmplus{} allocation is always \eqx{} (Theorem \ref{thm:LeximinIsEQx}). 

The following instance highlights the limitation of the \lmplus{} criterion under nonmonotone valuations. In particular, Table \ref{table:LeximinExample} that details the additive, objective valuations of two agents for three items: two goods, $g_1$, $g_2$ and a chore $c$. 

\begin{table}[!ht]
\centering
    \begin{tabular}{|c|c|c|c|}
    \hline
     & $g_1$ & $g_2$ & $c$ \\ \hline
     Agent 1 & $10$ & $1$ & $-1$ \\ \hline
     Agent 2 & $1$ & $100$ & $-1000$ \\ \hline     
    \end{tabular}
    \caption{Here, the \lmplus{} allocation is not \eqx{}.} 
    \label{table:LeximinExample}
\end{table}
In the given instance, consider bundles $A_1 = \{g_1, c\}$ and $A_2=\{g_2\}$. Under allocation $\alloc = (A_1, A_2)$, we have $v_1(A_1) = 9$ and $v_2(A_2) = 100$. Furthermore, in any other allocation, the minimum value among the two agents (i.e., the egalitarian welfare) is non-positive. Hence, for the instance at hand, $\alloc$ is the only \lmplus{} allocation.  

However, a chores violation exists under $\alloc$, since $v_2(A_1 \setminus \{c \}) = 10 < v_2(A_2)$. Therefore, allocation $\alloc$ is not \eqx{}. 
\end{document}